\documentclass[acmsmall]{acmart}
\pdfoutput=1
\usepackage{acm-ec-23}
\usepackage{booktabs} % For formal tables
\usepackage[ruled]{algorithm2e} % For algorithms

\SetAlFnt{\small}
\SetAlCapFnt{\small}
\SetAlCapNameFnt{\small}
\SetAlCapHSkip{0pt}
\IncMargin{-\parindent}

\usepackage{amsmath,amsfonts,amsthm}
\usepackage{mathtools}
\usepackage{thmtools}
\usepackage{algorithmicx}
\usepackage{nicefrac}
\usepackage{subcaption}
\usepackage{multirow}
\usepackage{comment}

\declaretheorem{example}
\declaretheorem{theorem}
\declaretheorem[sibling=theorem]{lemma}

\declaretheorem[sibling=theorem]{fact}

\theoremstyle{definition}
\declaretheorem{definition}
\declaretheorem{remark}

\DeclareMathOperator*{\argmax}{arg\,max}

\newcommand{\eps}{\epsilon}
\newcommand{\vc}{\boldsymbol}

\newcommand{\calM}{\mathcal{M}}
\newcommand{\calR}{\mathcal{R}}
\newcommand{\calU}{\mathcal{U}}

\newcommand{\neighbors}{\mathcal{N}}
\newcommand{\PPAD}{\textsf{PPAD}}
\newcommand{\FP}{\textsf{FP}}
\newcommand{\E}{\textsf{E}}

\usepackage{enumitem}% http://ctan.org/pkg/enumitem

\title[Practical algorithms and experimentally validated incentives for equilibrium-based fair division (A-CEEI)]{Practical algorithms and experimentally validated incentives for equilibrium-based fair division (A-CEEI)}
  
\author{Eric Budish}
\authornote{Disclosure: Cognomos Inc. implements the technology described in this paper as part of its commercial course allocation product, Schedule Scout. Budish and Othman are founding directors of Cognomos Inc.}
\email{eric.budish@chicagobooth.edu}
\affiliation{\institution{University of Chicago}\orcid{0000-0002-1483-8491}
\country{USA}}
\author{Ruiquan Gao}
\email{ruiquan@cs.stanford.edu}
\affiliation{\institution{Stanford University}\orcid{0009-0006-9837-8598}\country{USA}}
\author{Abraham Othman}
\authornotemark[1]
\email{abrahamo@wharton.upenn.edu}
\affiliation{\institution{University of Pennsylvania}\orcid{0000-0001-7992-4916}\country{USA}}
\author{Aviad Rubinstein}
 \authornote{Supported by NSF CCF-2112824, and a David and Lucile Packard Fellowship.}
\email{aviad@cs.stanford.edu}
\affiliation{\institution{Stanford University}\orcid{0000-0002-6900-8612}\country{USA}}
\author{Qianfan Zhang}
\email{qianfan@princeton.edu}
\affiliation{\institution{Princeton University}\orcid{0000-0003-3737-1545}\country{USA}}



\begin{abstract}
    Approximate Competitive Equilibrium from Equal Incomes (A-CEEI) is an equilibrium-based solution concept for fair division of discrete items to agents with combinatorial demands. In theory, it is known that in asymptotically large markets:
    \begin{itemize}
        \item For incentives, the A-CEEI mechanism is Envy-Free-but-for-Tie-Breaking (EF-TB), which implies that it is Strategyproof-in-the-Large (SP-L).
        \item From a computational perspective, computing the equilibrium solution is unfortunately a computationally intractable problem (in the worst-case, assuming $\PPAD\ne \FP$).
    \end{itemize}
    
    We develop a new heuristic algorithm that outperforms the previous state-of-the-art by multiple orders of magnitude. This new, faster algorithm lets us perform experiments on real-world inputs for the first time. We discover that with real-world preferences, even in a realistic implementation that  satisfies the EF-TB and SP-L properties, agents may have surprisingly simple and plausible deviations from truthful reporting of preferences. To this end, we propose a novel strengthening of EF-TB, which dramatically reduces the potential for strategic deviations from truthful reporting in our experiments. 

    A (variant of) our algorithm is now in production: on real course allocation problems it is much faster, has {\em zero} clearing error, and has stronger incentive properties than the prior state-of-the-art implementation.

\end{abstract}

\begin{document}

% Title page for title and abstract only.
% \begin{titlepage}

\maketitle

% \end{titlepage}

% Paper body
\section{Introduction}
\label{sec:intro}

Competitive Equilibrium from Equal Incomes (CEEI)~\cite{Fol67,Var74,TV85} is an attempt to leverage the economic efficiency of market equilibria while preserving the \emph{ex post} fairness properties that come from equal incomes.
For complex preferences, however, CEEI does not necessarily exist.

Budish~\cite{Budish2011} developed a relaxation of CEEI, {\em Approximate CEEI (A-CEEI)}. He showed that an approximate equilibrium from approximately equal incomes always exists.
Budish~\cite{Budish2011} also showed if the perturbations to agents' incomes (henceforth {\em budgets}) are chosen at random, the mechanism satisfies an Envy-Free-but-for-Tie-Breaking (EF-TB) property, which by~\cite{AB18} implies that it is Strategyproof in the Large (SP-L); i.e.~as the number of agents tends to infinity, the mechanism becomes approximately strategyproof.

Budish's original theoretical work invoked a fixed-point theorem to prove existence. That makes it inherently non-constructive. It was shown in~\cite{OthmanPR16} that finding an A-CEEI is \PPAD-complete, even if we allow constant budget inequality~\cite{Rub18}. A-CEEI is therefore similar to many other economic equilibria whose existence rely on non-algorithmic proofs, the most famous of which is the Nash equilibrium~\cite{DaskalakisGP09,ChenDT09}. 

Despite this theoretical infeasibility, a heuristic algorithm exists that solves the problem adequately in practice~\cite{BudishCKO17}. %This algorithm uses the economic principle of tatonnement local search to progressively refine the prices assigned to courses.  %Local search is embarrassingly parallel and in practice that algorithm is run over many search starts on many compute servers.
%\Aviad{Let's leave the high level description of the algorithm to Section~\ref{sub:contribution1}} \Abe{OK!}
%
The existing heuristic algorithm makes A-CEEI a practical solution concept for settings that require efficiency and fairness but for which the use of real money is impractical or repugnant~\cite{Roth07}. The setting  where A-CEEI has seen the most practical application is in the allocation of courses to students.
 Student preferences are often quite complex in course allocation. Students typically demand many courses and individual courses could be complements or substitutes depending on the bundle of other courses in a student's schedule. Course allocation---particularly in professional schools---also tends to be a challenging allocation problem, as the most popular ``star courses" tend to have far more demand than supply.
%
%Prior to our work the heuristic algorithm provided a functional solution, but it was (i) slow and (ii) introduced potential for manipulability (see discussion in Subsection~\ref{sub:zero}).

In this paper we explore the course allocation problem experimentally using real data from properly motivated student preferences over schedules of professional school courses, using data from the commercial implementation of A-CEEI fielded by Cognomos.

%\Aviad{I'm moving discussion of second and third stage to Subsection~\ref{sub:zero-error}} \Abe{OK!}

%\Aviad{How should we call the different EF-TB constraints? Regular and Strong? Classic and Robust? Other suggestions?}

%\Aviad{To do throughout the paper, especially Case Studies section: anonymize / remove school names, etc}

\subsection{Contribution I: A (much) faster heuristic algorithm for computing A-CEEI}\label{sub:contribution1}

\begin{figure}[t]
\centering
\includegraphics[width=.6\textwidth]{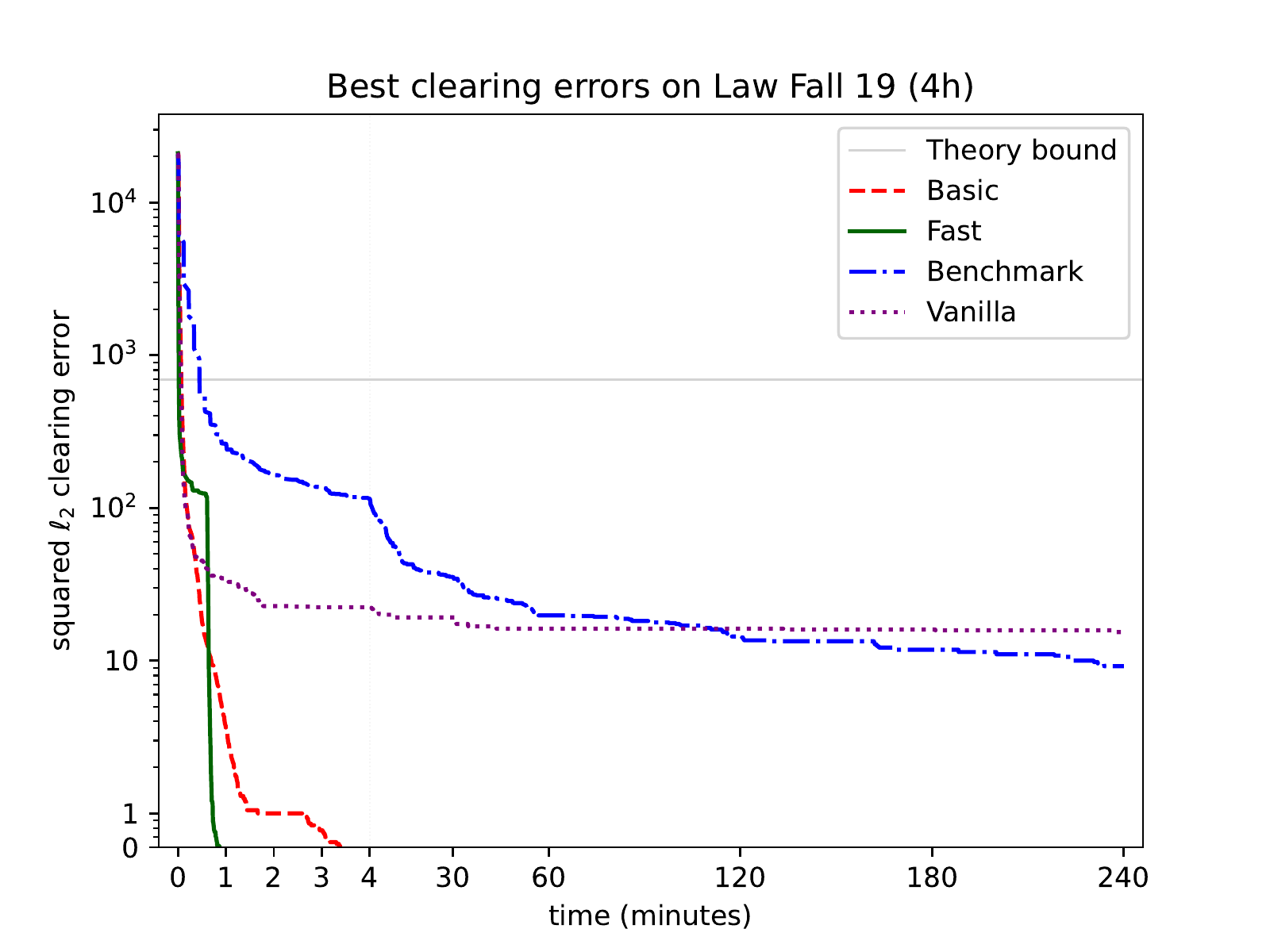}
\caption{For each of the algorithms tested and time $t$ (up to 4-hours), we plot the best clearing error obtained by the algorithm up to time $t$. Note that our final algorithm (Fast) finds a zero-clearing-error solution after less than 1 minutes, while previous commercial state-of-the-art (Benchmark) does not converge after 4 hours. See Section~\ref{sec:perform} for details.}
\label{fig:progress-of-algos-error-4h}
\end{figure}

Our first contribution is an improved heuristic algorithm for computing A-CEEI. Our algorithm outperforms the commercial state-of-the-art by several orders of magnitude in both the quality of produced solution as well as runtime on real course allocation problems, see e.g.~Figure~\ref{fig:progress-of-algos-error-4h}.

%\footnote{To be precise we improve over the {\em previous} commercial state-of-the-art --- a variant of the algorithm we developed in this work is now deployed by Cognomos.} on all fronts, sometimes by multiple orders of magnitude, see e.g.~Figure~\ref{fig:TBD} \Aviad{TBD: move figure to intro }.
% Abe doesn't think this is a good idea. No figures in the Intro, don't even really need the forward pointer.

We highlight some of our algorithmic findings here, with details in Sections~\ref{sec:algo} and \ref{sec:perform}. Further details about the previous state-of-the-art algorithm (henceforth {\em benchmark}), can also be found in academic publications~\cite{OthmanSB10, BudishCKO17}. 
At a high level, both our algorithm and the benchmark, perform variants of the classic tatonnement algorithm: increase the price of over-demanded courses, and decrease the prices of under-demanded courses.

One of the main technical innovations introduced in the benchmark algorithm was the use of individual price adjustments: on some iterations, the algorithm can make a larger improvement on the clearing error by adjusting the price of a single course rather than all the courses simultaneously. Based on extensive experiments on randomly generated data (real data from students bidding was not yet available at the time),~\cite{OthmanSB10} reported that the algorithm was much more likely to find a solution with acceptable clearing error when mixing individual price adjustments and full tatonnement updates. 
Our first algorithmic insight is to largely {\em reverse this finding of~\cite{OthmanSB10}}. We observe that on real world instances, it is in fact much more efficient to use only full tatonnement updates, even when they locally increase the clearing error. We discuss evidence, explanations, and caveats of  this finding in Section~\ref{sec:perform}.

Our second (and perhaps more interesting) algorithmic insight focuses on the ``end game'', when the algorithm is already close to a reasonable clearing error. Here, we show that making tiny-but-cleverly-optimized perturbations to the {\em budgets} (rather than prices) can quickly lead to the holy grail of zero clearing error. Before going into the details, we remark that (i) in the practical benchmark implementation, students budgets are already perturbed at random (and with larger perturbations); (ii) our insight  is inspired directly by the non-algorithmic existence proof in~\cite{Budish2011}, which perturbs the students' budgets twice: once to guarantee a desirable fixed point of the tatonnement correspondence, and a second time to break ties between marginal students at the fixed point. 

To describe the budget perturbations, let's start from the end: we would like to return a vector of course prices and almost-equal student budgets and (ideally) zero clearing error. Towards this, at each iteration of the tatonnement the algorithm solves an (NP-hard but fast in practice) integer program to look for a small budgets perturbation that will nudge students' demand bundles to zero the overall market clearing error. More generally, at each iteration we can find the {\em optimal budgets perturbation}, aka the one that minimizes the clearing error. Since our ultimate goal is to find a price vector that works well with an optimal budget perturbation, we use the clearing error with respect to the optimal budget perturbation to perform the next iteration of tatonnement.

While the idea of using optimal budget perturbations is extremely effective algorithmically, we have to be careful not to open opportunities for manipulability: although budget increases are very small, doing it based on student's reported course preferences could open an opportunity for strategic reporting. To this end, we encode in our integer program the same EF-TB constraints that guarantee the SP-L (Strategyproof in the Large) condition for the original A-CEEI mechanism. Although this makes the integer program larger, we can solve it very efficiently (see Section~\ref{sub:fast-EF-TB} for some optimizations).

%TBD: discuss the new algorithm and comparisons with previous algorithm, other approaches.

\subsection{Contribution II: Empirically evaluating the incentives of A-CEEI}

In theory, the A-CEEI mechanism has a very desirable property: it is {\em Strategyproof in the Large (SP-L)}; i.e.~the expected utility any student can gain by misreporting their true preferences diminishes as the number of students goes to infinity~\cite{AB18}%
\footnote{The informal intuition is that students who are ``price-takers'', i.e.~they regard the prices of courses as set ``by the market'', receive the optimal schedule they can afford and thus have no incentive to misreport their preferences.}. 
What does this theory guarantee for realistic schools with a finite number of students? 
Unfortunately, the answer is not much: the formal SP-L convergence guarantee (~\cite[Theorem 1]{AB18}) requires the number of students to be much larger than the number of possible types. For the general A-CEEI mechanism%
\footnote{In practice, students' reporting language is often more restricted, but still the number of students is always much smaller than the number of possible types.}, a type consists of a ranking of all possible schedules; the number of possible schedules is exponential in the number of courses, and the number of ways they can be ranked adds another layer of exponentiation. In other words, until the number of students is {\em doubly-exponential} in the number of courses, the formal convergence results for SP-L mechanisms are meaningless. 

The situation is further complicated by the computational intractability of the A-CEEI problem. %
A natural approach for dealing with the large number of possible deviations is to directly analyze the way students could manipulate the algorithm's choice of equilibrium. This has been fruitful in analyzing simple algorithms~\cite{Vitercik21-thesis,AKLLM20}. But because of the \PPAD-hardness of the A-CEEI problem we resort to a highly nontrivial heuristic algorithm  (see Contribution I). Theoretically analyzing how a possible misreport of preferences would affect the trajectory in price space taken by this algorithm seems far from tractable.

Going beyond intractable theoretical guarantees on students' incentives, our approach is to empirically evaluate them. 
Our main question in this part is:
\begin{quote}
    In practice, can students gain in the A-CEEI mechanism by manipulating their reported preferences?
\end{quote}

Here, it is paramount that we have our novel fast heuristic algorithm: Previously, it was barely possible to compute one allocation, so re-computing allocations for each possible deviation was completely out of the question. Anecdotally, this is the reason why the computational exploration of the same issue in Budish's original paper was limited to  tiny examples with 2 students and 4 courses~\cite[Footnote 31]{Budish2011-old}.

Still, no matter how fast our equilibrium computation algorithm, we cannot possibly enumerate all possible deviations (again, in theory their number is doubly exponential in the number of courses). However, it is reasonable to expect that computationally- and informationally-bounded students also cannot test every possible deviation. We thus model a strategic student using a simple hill-climbing algorithm that adjusts the single course weights starting from the original truthful report. We also consider different restrictions on the student's information about the market, as detailed in Definition~\ref{def:manipulation}.

In Section~\ref{sec:manip} we use our manipulation-finding algorithm in combination with our fast A-CEEI finding algorithm to explore the plausibility of effective manipulations for students bidding in A-CEEI. Originally, we had expected that since our mechanism satisfies the EF-TB and SP-L properties, it would at least be practically strategyproof --- if even we don't really understand the way our algorithm chooses among the many possible equilibria, how can a student with limited information learn to strategically bid in such a complex environment?

Indeed, in 2 out of 3 schools that we tested, our manipulation-finding algorithms finds very few or no statistically significant manipulations at all. However, when analyzing the 3rd school,  we stumbled upon a simple and effective manipulation for (the first iteration of) our mechanism. We emphasize that although the manipulation is simple in hindsight, in over a year of working on this project we failed to predict it by analyzing the algorithm --- {\em the manipulation was discovered by the algorithm}. 

Inspired by this manipulation, we propose a natural strengthening of envy-free (discussed below), which we call {\em contested-envy free}. We encode the analogous {\em contested EF-TB} as a new constraint in our algorithm (specifically, the integer program for finding optimal budget perturbations). Fortunately, our algorithm is still very fast even with this more elaborate constraint. And, when we re-run our manipulation-finding experiments, we observe that contested EF-TB significantly reduces the potential for manipulations in practice.

\subsection{Contribution III: Contested Envy Free (but for Tie Breaking)}
In this section we gradually build towards our new notions of contested-envy free, and contested EF-TB.
We begin with the basic notion of {\em Envy-Free (EF)}: an allocation is said to be envy free if no student $i$ prefers the schedule allocated to another student $i'$.

In the course allocation problem, due to the challenges of integrality constraint and combinatorial demand, EF allocations rarely exist, but A-CEEI allocations are guaranteed to satisfy important relaxations of EF such as EF-TB (discussed below). 
In contrast, contested-EF is a {\em strengthening} of EF. To motivate the distinction between EF and contested-EF, consider the following anecdote. (This anecdote is for illustration purposes only; in Section~\ref{sec:anecdotes} we discuss examples of courses and students derived from manipulations found by our algorithm on instances with real preferences.)

\begin{example}[Contested envy free]\hfill

Eric drives a Honda and Mohammad drives a Porsche. Eric would rather have Mohammad's Porsche than his Honda.  However, Eric doesn't {\em envy} Mohammad, because his kids are in his Honda, and he wouldn't trade the bundle of $\{\text{the Honda and his kids}\}$ for Mohammad's Porsche (with no kids). 
\end{example}

Eric loves his own kids, but nobody else would want them in their car; we thus say that they're {\em uncontested} for understanding the envy between Eric and Mohammad. 
Formally, in the the specific context of A-CEEI%
\footnote{More generally when prices aren't available it is natural to extend the notion of ``uncontested'' to capture under-demanded goods.}
we say that contested are the goods with strictly positive price, and goods with zero price are uncontested.
Considering uncontested goods for the purposes of determining envy is an obvious source of incentive issues: Eric can always report a low value for his kids to claim to envy Mohammad's allocation. Because they're uncontested, Eric is still guaranteed to have them in any Pareto optimal allocation.

If we restrict our attention to {\em contested} goods, Eric would indeed rather trade his Honda for Mohammad's Porsche.
More generally, we say that Eric {\em contested-envies} Mohammad if Eric prefers any subset of $\{\text{Mohammad's allocation}\} \cup \{\text{uncontested goods}\}$ over his own allocation. Notice that contested EF is a strengthening of EF. 

As mentioned before, EF allocations rarely exist in the course allocation problem. 
Azevedo and Budish~\cite{AB18} relax the notion of EF to allow {\em tie-breaking} (EF-TB): The students are ranked at random%
\footnote{In practice a combination of seniority and random ranking may be used.}, and the allocation is said to satisfy EF-TB if no student $i$ envies any lower-ranked students. (However, $i$ may envy higher-ranked students.)
EF-TB is not a very satisfying fairness criterion%
\footnote{For fairness,~\cite{Budish2011} introduced the notion of Envy-Free-up-to-1-good (EF1) and proved that it is satisfied by A-CEEI.}, but it does imply strategyproof-in-the-large (SP-L)~\cite{AB18}. 

A-CEEI allocations satisfy EF-TB when the budgets are assigned at random because no student can envy another student with a lower budget (if $i'$'s budget is lower, then $i'$'s schedule must also be affordable for $i$). When we introduce optimal budget perturbations, we simply encode  EF-TB as a constraint in our perturbation-finding integer program: if $i'$'s initial tie breaker is lower than $i$'s, then the EF-TB constraint is that $i$ cannot envy $i'$ in the perturbed economy. While in theory this approach satisfies SP-L, in practice it can open the door to manipulations; see Sections~\ref{sec:anecdotes} and~\ref{sec:manip} for discussion and evidence from experiments.

We can generalize EF-TB to {\em contested EF-TB} in the natural way: no student can contested-envy a lower ranked student. 
We henceforth refer to the EF-TB criterion as {\em classic EF-TB} to distinguish it from contested EF-TB. 

Note that contested EF-TB is a strengthening of classic EF-TB, so it also implies SP-L. Observe also that if all the budgets respect the random TB rule, then any A-CEEI allocation is contested EF-TB. In our new algorithm, because of the optimal budget perturbations, the budgets  may not respect the original TB order; instead, we encode the contested EF-TB constraint in the integer program that optimizes over budget perturbations.  

%In Section~\ref{sec:anecdotes} we describe \Aviad{TBD: continue...}. 
In Section~\ref{sec:manip} we bring a quantitative analysis of our manipulation-finding algorithm: we give evidence from experiments on real students bids, suggesting that when we enforce contested EF-TB, it is fairly hard to find successful manipulations.

\subsection{Limitations of our approach}
\paragraph{Modeling students' valuations} Throughout the paper we take students' reported preferences as their true valuations. In practice, students do not perfectly report their full preferences~\cite{BK22}. In fact, students usually have a limited interface, e.g. they may rate courses as ``favorite, great, good, fair'', and there is a hand-tuned formula converting  those ratings to utilities. Concurrent work by~\cite{SZWS23} focuses on the orthogonal direction of better eliciting and modeling of students' preferences using neural networks. Equilibrium computation is a major bottleneck of their approach, and we leave it to future work to see if our respective algorithms can be combined effectively.

\paragraph{Variance across markets} We are very fortunate to be able to test our algorithms on real data from a few different programs shared with us by Cognomos. There seem to be large variance between instances from different programs. In particular, in one school we find significant manipulations for the no-EF-TB and classic-EF-TB constraints, whereas in others those variants were also hard to manipulate. Running times also vary greatly between instances. However, two trends are consistent between all the instances we tested: (i) our algorithm is much faster than the previous state of the art, and (ii) contested EF-TB seems to have desirable incentive properties. We plan to test on more datasets as they become available. 

\paragraph{Limitations of the manipulation finding algorithm}
Our algorithm for discovering profitable manipulations is a highly imperfect surrogate to the real question of {\em can real students find robust manipulations?} On one hand, students can come up with more complicated manipulations than the ones considered by our algorithm; on the other hand, the algorithm has more information than any single student would normally have.  A particular issue is that we run our experiment for exploring profitable deviations with very small sample sizes: even with a very fast algorithm and restriction to very simple manipulations, we have to restrict our experiment to a as few as five samples, which is quite noisy. (We later validate every candidate manipulation with at least 100 iterations.) To make sure that our manipulation-finding algorithm still makes sense, we benchmark it on the HBS mechanism which is known to be manipulable~\cite{BC12}; indeed we find significant manipulations are possible with HBS on all instances.

\subsection{Conclusion}
In this work, we give a significantly faster algorithm for computing A-CEEI. Kamal Jain's famous formulation ``if your laptop cannot find it then neither can the market''~\cite{AGT-book} %[CITE: AGT Book. I can't actually find this quote anywhere in a paper?] \Aviad{Abe, if your laptop can't find the citation... ;)} LOL
 was originally intended as a negative result, casting doubt on the practical implications of many famous economic concepts because of their worst-case computational complexity results. Even for course allocation, where a heuristic algorithm existed and worked in practice, Jain's formulation seemed to still bind, as solving A-CEEI involved an intense day-long process with a fleet of high-powered cloud servers operating in parallel.
The work detailed in this paper has significantly progressed what laptops can find: even the largest and most challenging real course allocation problems we have access to can now be solved in under an hour on a commodity laptop.
%\Aviad{@ruiquan: is it a few hours? Can't we do it in minutes?}
%\Abe{I believe Columbia Biz School, which is really REALLY hard, still takes hours?}
%\Ruiquan{Yes, the $(0,0.04)$-CEEI can now be found in 1h for the 10 randomly sampled vectors of initial budgets. We use Gurobi and a lot of speedups shown in Section 3.} 
% \Abe{I just made this "an hour"}

This significant practical improvement suggests that the relationship between prices and demand for the course allocation problem---and potentially other problems of economic interest with complex agent preferences and heterogeneous goods---may be much simpler than has been previously believed and may be far more tractable in practice than the worst-case theoretical bounds. Recalling Jain's dictum, perhaps many more market equilibria can be found by laptops---or, perhaps, Walras's original and seemingly naive description of how prices iterate in the real world may in fact typically produce approximate equilibria.

Our fast algorithm also opens the door for empirical research on A-CEEI, because we can now solve many instances and see how the solution changes for different inputs. We took it in one direction: empirically investigating the incentives properties of A-CEEI for the first time. For course allocation specifically, this faster algorithm opens up new avenues for improving student outcomes through experimentation. For instance, university administrators often want to subsidize some group of students (e.g., second-year MBA students over first-year MBA students), but are unsure how large of a budget subsidy to grant those students to balance equity against their expectations. Being able to run more assignments with different subsidies can help to resolve this issue.

%We also highlight that in contrast to previous work, out algorithm can actually find a {\em zero}-clearing-error in all the instances we considered. Below, we explain one reason why this is important in practice.  

\begin{remark}[Zero vs small clearing error]\label{remark:zero-error}

We highlight that our algorithm is not only fast - it also finds allocations with {\em zero} clearing error. Even the non-algorithmic existence proof of~\cite{Budish2011} only guarantees a {\em small} clearing error.

While the previous heuristic algorithm was able to find adequate allocations in practice, it introduces some additional potential manipulability. That algorithm was a three-stage process. The first stage finds an approximate equilibrium, which equilibrium will tend to have both undersubscription in positive-price courses as well as oversubscription. This is the approximate equilibrium guaranteed to exist by~\cite{Budish2011}. The second stage progressively increases course prices to eliminate all oversubscription. This tends to increase total clearing error but makes the solution implementable, since all of the error comes from underallocating seats in valuable courses.  The final stage is a ``backfill'' process, where students are sequentially allocated extra budget and allowed to spend that budget on courses that are undersubscribed.

While the backfill process substantially reduces the deadweight loss of computed assignments it may have problematic incentive and fairness properties. Students who get first shot at spending that extra budget in the backfill may be able to add an excellent course to their schedule. In particular, the backfill process is not known to satisfy properties like (contested) EF-TB or SP-L. Observe, however, that the backfill process is only necessary because of the clearing error found in the original approximate equilibrium. In contrast, so far our new algorithm has found A-CEEI with {\em zero} clearing error on every instance it has encountered. This completely obviates the need for the second and third stages of~\cite{BudishCKO17}: with no clearing error, there are no seats that need to be backfilled.
\end{remark}

\begin{remark}[Social welfare]
Although our main focus is on improving the algorithmic efficiency and incentives guarantees of the A-CEEI mechanisms, in Appendix~\ref{app:econeff} we compare the social welfare of our algorithm and the previous approach. Although the results aren't as decisive as on other metrics, we observe that our algorithm tends to give better allocations in terms of (utilitarian and Nash) social welfare.
\end{remark}

Discussion of additional related work can be found in Appendix~\ref{app:related}.

%\section{Preliminary}
\label{sec:prelim}

\section{Preliminaries}

\begin{definition}[The course allocation market]
A {\em course allocation market} $\big(\vc{u}=(u_i)_{i=1}^n,\vc{c}=(c_j)_{j=1}^m\big)$ consists of:
\begin{itemize}
    \item $m$ courses, where each course $j\in [m]$ has an integral amount of capacity $c_j$;
    \item $n$ students, where each student $i\in [n]$ has a utility function $u_i$ over each course bundle.
    %\Qianfan{Should we highlight that $u_i$ might not be monotone?}
\end{itemize}
\end{definition}

\begin{definition}[Allocation, excess demand, and market-clearing error]
Fix a market $(\vc{u},\vc{c})$, course prices $\vc{p}=(p_j)_{j=1}^m$, and student budgets $\vc{b}=(b_i)_{i=1}^n$, the allocation function $\vc{a}=(\vc{a}_i)_{i=1}^n$ is defined as
$$\vc{a}_i(\vc{u},\vc{p},\vc{b})= \argmax_{\vc{x} \in 2^{[m]}, \vc{p}\cdot\vc{x}\le b_i} u_i(\vc{x}).$$
We further define the {\em excess demand} function $\vc{z}=(z_j)_{j=1}^m$ as 
\begin{equation*}
    z_j(\vc{u},\vc{c},\vc{p},\vc{b}) = \sum_{i=1}^n a_{ij}(\vc{u},\vc{p},\vc{b}) - c_j,
\end{equation*}
and the {\em clipped excess demand} function $\vc{\tilde{z}}=(\tilde{z}_j)_{j=1}^m$ as
\begin{equation*}
    \tilde{z}_j(\vc{u},\vc{c},\vc{p},\vc{b}) = \begin{cases}
    z_j(\vc{u},\vc{c},\vc{p},\vc{b}) & \text{if $p_j>0$,}\\
    \max\{0,z_j(\vc{u},\vc{c},\vc{p},\vc{b})\} & \text{if $p_j=0$.}
    \end{cases}
\end{equation*}
And we define the {\em market-clearing error} as $\|\vc{\tilde{z}}(\vc{u},\vc{c},\vc{p},\vc{b})\|_2$.
\end{definition}

\begin{definition}[Approximate competitive equilibrium from equal incomes (A-CEEI)]
For constant $\alpha,\beta>0$ and market $(\vc{u},\vc{c})$, we say a pair of prices and budgets $(\vc{p},\vc{b})$ forms an {\em $(\alpha, \beta)$-CEEI} if there is
\begin{itemize}
    \item small market-clearing error: $\|\vc{\tilde{z}}(\vc{u},\vc{c},\vc{p},\vc{b})\|_2 \le \alpha$;
    \item small budget perturbation: $b_i \in [1,1+\beta] \;\;\; \forall i$. % $\|\vc{b}-\vc{1}\|_\infty \le \beta$.
\end{itemize}
\end{definition}

%\begin{definition}[A-CEEI algorithm] \Aviad{I find this definition confusing. Do we really need it?}
%For constant $\alpha,\beta>0$, a deterministic algorithm $\calA(\vc{u},\vc{c},\vc{b_0})$ is called an $(\alpha,\beta)$-CEEI algorithm, if on the input of the market $(\vc{u},\vc{c})$ and an initial budget $\vc{b_0}$, it outputs prices and final budgets $(\vc{p},\vc{b})$ such that $(\vc{p},\vc{b})$ forms an $(\alpha,\beta)$-CEEI.

%Furthermore, we denote the resulting allocation by $\vc{A}(\vc{u},\vc{c},\vc{b_0})$, i.e.,
%$$\vc{A}_i(\vc{u},\vc{c},\vc{b_0}) = \vc{a}_i(\vc{u},\vc{p},\vc{b})$$
%where $(\vc{p},\vc{b})=\calA(\vc{u},\vc{c},\vc{b_0})$.
%\end{definition}

\subsection{(Contested) Envy-Free-but-for-Tie-Breaking}\label{sub:prelim-EFTB}

We formally define the notions of EF-TB and contested EF-TB. For both, it is important to make the distinction between the initial budgets $\vc{b_0}$ that are determined exogenously (e.g.~at random), and the final budgets $\vc{b}$ which may also depend on the reported preferences. In both cases, students' initial budgets play a second role in determining the direction in which envy is allowed; in particular the initial budgets are assumed to be distinct (hence ``tie-breaking'').

\begin{definition}[Envy-Free-but-for-Tie-Breaking (EF-TB)]
\label{def:ef-tb}
Given an initial budget $\vc{b_0}$, for a market $(\vc{u},\vc{c})$, price $\vc{p}$, and budget $\vc{b}$, the allocation $\vc{a}(\vc{u},\vc{p},\vc{b})$ is called {\em EF-TB with respect to budget $\vc{b_0}$}, if for all student $i,j\in [n]$ such that $b_{0,i} > b_{0,j}$, we have $u_i(a_i(\vc{u},\vc{p},\vc{b})) > u_i(S)$ for all bundle $S \subseteq a_j(\vc{u},\vc{p},\vc{b})$.

Furthermore, we say that an A-CEEI algorithm $A(\vc{u},\vc{c},\vc{b_0})$ is EF-TB, if for any market $(\vc{u},\vc{c})$ and any initial budget $\vc{b}_0$, the final allocation $\vc{A}(\vc{u},\vc{c},\vc{b_0})$ is always EF-TB with respect to $\vc{b_0}$.
\end{definition}

%\Qianfan{Here we describe ``EF-TB with subsets'' directly. Maybe we should explain more.}
%\Aviad{I think it'}

\begin{definition}[Contested Envy-Free-but-for-Tie-Breaking (Contested EF-TB)]
\label{def:strong-eftb}
Given an initial budget $\vc{b_0}$, for a market $(\vc{u},\vc{c})$, price $\vc{p}$, and budget $\vc{b}$, the allocation $\vc{a}(\vc{u},\vc{p},\vc{b})$ is called {\em contested EF-TB with respect to budget $\vc{b_0}$ and price $\vc{p}$}, if for all student $i,j\in [n]$ such that $b_{0,i} > b_{0,j}$, we have $u_i(a_i(\vc{u},\vc{p},\vc{b})) > u_i(S)$ for all bundle $S$ such that $S \subseteq a_j(\vc{u},\vc{p},\vc{b}) \cup \{k \in [m] : p_k=0\}$.

Furthermore, we say that an A-CEEI algorithm $A(\vc{u},\vc{c},\vc{b_0})$ is contested EF-TB, if for any market $(\vc{u},\vc{c})$ and any initial budget $\vc{b}_0$, the final allocation $\vc{A}(\vc{u},\vc{c},\vc{b_0})$ is always contested EF-TB with respect to $\vc{b_0}$ and $\vc{p}$.
\end{definition}

\subsection{Utility functions}
\label{sub:utility}

While the A-CEEI existence works for general ordinal preferences, we focus on the following restricted class of utility functions. This class is consistent with most utilities reported in practice, which are typically taken to be additively-separable utilities, (i) with a preference for schedules satisfying a minimum-number-of-course-units requirements, and (ii) subject to satisfying simple constraints, e.g.~timing and curriculum conflicts. Using the language of Operations Research, schedule validity is a hard constraint, while having a schedule meet a student's requirements is a soft constraint.  (The problem remains \PPAD-complete in the worst-case when restricted to this class; see Appendix~\ref{app:hard}.) 

%\Abe{ Do we need to have valid and req here at all?? Would it suffice to say to say that the utility functions are just efficiently representable in a mixed-integer program for each student?}
%\Aviad{We should explain this, but for the purpose of manipulability we assume that valid and req are predetermined by school so students only modify course weights. Also, after our old PPAD-hardness paper someone.. maybe even Eric? asked me if the problem remains PPAD-hard for the bidding language in practice. So it's nice to have that settled!}
Formally, every utility function $u$ can be described by a tuple $(\vc{w} \in \mathbb{R}^m, valid, req: 2^{[m]}\to \mathbb{N})$, such that for every possible bundle $\vc{x} \in 2^{[m]}$,
\begin{gather} \label{eq:original-utilities}
u(\vc{x})=\begin{cases} 
\vc{w}\cdot\vc{x} + B & valid(\vc{x})=1, req(\vc{x})=1\\
\vc{w}\cdot\vc{x} & valid(\vc{x})=1, req(\vc{x})=0\\
-\infty & valid(\vc{x})=0
\end{cases}
\end{gather}
where $B$ is some large number such that $B>\|\vc{w}\|_1$.
The function $valid$ and $req$ also follow some structures so that they can be efficiently represented (e.g., in a mixed-integer program).

% While the A-CEEI existence works for general ordinal preferences, we focus on the a restricted class of utility functions, where students choose weights for courses, but there are complicated combinatorial constraints imposed by the school. This class is consistent with most utilities reported in practice and can be efficiently representable by mixed-integer programs (see Appendix~\ref{app:utility}). The problem remains \PPAD-complete in the worst-case when restricted to this class; see Appendix~\ref{app:hard}.

%We should first specify the class of utility functions we consider.

%In the experiments we considered the following instances:
%\begin{enumerate}

\section{Our algorithm}
\label{sec:algo}
In this section we describe our fast heuristic algorithm for computing A-CEEI. 
We begin with a description of the basic algorithm (Subsection~\ref{sub:basic}). 
We then move to describe some further optimizations that we found helpful when solving larger instances (Subsection~\ref{sub:Shortcuts}).

\subsection{Our Basic Algorithm}\label{sub:basic}
All our algorithms take as inputs the students' reported preferences $\vc{u}$, the course capacities $\vc{c}$, and initial budgets $\vc{b_0}$ that are determined at random, sometimes with a bonus for seniority. 
Our basic algorithm proceeds as follows: At each iteration, the algorithm looks for the optimal budget perturbation given current prices; then it computes the market clearing error for this optimal budget perturbation; finally, it updates the prices according to tatonnement rule. The algorithm terminates when it reaches zero clearing error\footnote{It is also a good idea to enforce a time limit as a solution with zero clearing error is not even guaranteed to exist, but in practice we managed to find solutions with zero clearing error on all the instances we encountered.}.
The pseudocode is given in Algorithm~\ref{alg:tatonnement}.

%Following is our deterministic algorithm for finding an A-CEEI with (strong) EF-TB property.

\begin{algorithm}[t]
\caption{Find an A-CEEI with (contested) EF-TB property}
\label{alg:tatonnement}
\begin{algorithmic}
\State \textbf{Inputs:} students' utility functions $\vc{u}$, course capacities $\vc{c}$, initial budgets $\vc{b_0}$
\State \textbf{Outputs:} final prices $\vc{p}^*$ and budgets $\vc{b}^*$
\State \textbf{Parameters:} step size $\delta$, maximum budget perturbation $\epsilon$, type $t$ of the EF-TB constraint used ($0$ for no EF-TB constraint, $1$ for EF-TB constraint, and $2$ for contested EF-TB)
\State \textbf{Algorithm:}
\begin{enumerate}
  \item Let $\vc{p} \gets \vc{0}$.
  \item \label{step-np} {\bf $\epsilon$-budget perturbation}: find budgets $\vc{b}$ such that the market-clearing error $\|\vc{\tilde{z}}(\vc{u},\vc{c},\vc{p},\vc{b})\|_1$ is minimized under the following constraints:
  \begin{enumerate}
    \item The maximum perturbation $\|\vc{b}-\vc{b}_0\|_{\infty} \le \epsilon$;
    \item Allocation $\vc{a}(\vc{u},\vc{p},\vc{b})$ is EF-TB with respect to $\vc{b}_0$ if $t=1$, or contested EF-TB with respect to $\vc{b}_0$ and $\vc{p}$ if $t=2$.
  \end{enumerate}
  Furthermore, we shall use $\|\vc{b}\|_1$ as the tie-breaker to guarantee the uniqueness of the solution, i.e., always picking $\vc{b}$ with minimum $\|\vc{b}\|_1$ among all optimal solutions.
  \item If $\|\vc{\tilde{z}}(\vc{u},\vc{c},\vc{p},\vc{b})\|_2=0$, terminate with $\vc{p}^*=\vc{p}, \vc{b}^*=\vc{b}$.
  \item Otherwise, update $\vc{p}\gets \vc{p}+\delta\vc{\tilde{z}}(\vc{u},\vc{c},\vc{p},\vc{b})$, then go back to step 2.
\end{enumerate}
\end{algorithmic}
\end{algorithm}

\subsubsection*{Computing the $\epsilon$-budget perturbation.} To compute the $\epsilon$-budget perturbation in the second step, we should optimize among all possible budget perturbations so that the resulting clearing error is minimized. Observe that for any fixed price vector, the demand of any student can only change on some budgets that are the sum of some prices; and the sum of prices is always a multiple of the fixed step size $\delta$. 

Therefore, we can always partition the interval $[b_{0i}\pm \epsilon]$ of Student $i$'s possible budgets into $k_i\leq \frac{2\epsilon}{\delta} + 1$ sub-intervals $(\underline{b_{i\ell}}, \overline{b_{i\ell}}]$, such that  $i$'s demand bundle $\vc{a_{i\ell}}$ is constant on each sub-interval:
\begin{align}
    \label{eqn:array-of-demands}
    \big\{(\vc{a_{i\ell}}, \underline{b_{i\ell}}, \overline{b_{i\ell}}): \text{student $i$'s demand is $\vc{a_{i\ell}}$ for every budget in $[\underline{b_{i\ell}}, \overline{b_{i\ell}}]$}\big\}_{\ell=1}^{k_i},
\end{align}
where $\underline{b_{i\ell}}, \overline{b_{i\ell}}$s are multiples of $\delta$ in $b_{0i}\pm \epsilon$. %and the union of $[\underline{b_{ij}}, \overline{b_{ij}}]$ consists of all multiples of $\delta$ in $b_{0i}\pm \epsilon$. Note that it implies $k_i\leq \frac{2\epsilon}{\delta} + 1$, which is polynomial. 

Once we compute these arrays, we can solve for  the optimal $\epsilon$-budget perturbation using the following integer linear program:
\begin{align}
    &\tag{\textsc{Budget-Perturb-ILP}} \label{eqn:budget-perturb-ilp}
    \\
     &\min \quad  \|\vc{z}\|_1 & && \text{(Minimize clearing error)}
    \notag
    \\
    \text{s.t.}\quad & \sum_{i\in[n]} \sum_{\ell \in[k_i]} x_{i\ell }\cdot  a_{i\ell j} = c_{j}+z_{j} & \forall j \in [m], p_{j}>0 && \text{(Clearing error: $p_j>0$)}
    \notag
    \\
    & \sum_{i\in[n]} \sum_{\ell \in[k_i]} x_{i\ell }\cdot  a_{i\ell j} \leq c_{j}+z_{j} & \forall j \in [m], p_{j}=0 && \text{(Clearing error: $p_j=0$)}
    \notag
    \\
    & \sum_{\ell \in [k_i]} x_{i\ell} = 1 & \forall i\in [n] && \text{(1 schedule per student)}
    \notag
    \\
    & x_{i\ell}\in \{0,1\} & \forall i\in [n], \ell \in [k_i] && \text{(Integral allocations)}
    \notag
    % \\
    % & x_{ij}\geq x_{i'j'} & \forall i, i' \in [n], j\in [k_i], j'\in [k_{i'}], u_i(\vc{a_{ij}}) < u_i(\vc{a_{i'j'}})
    % \\
    % & z_{\ell}\in \mathbb{R} & \forall \ell\in [m]
\end{align}
We add the following constraint to ensure (contested) EF-TB: For any student $i, i'\in [n]$ such that $i$'s priority is higher than $i'$ (i.e. $b_{0i}>b_{0i'}$), and any $\ell\in [k_i], \ell'\in [k_{i'}]$, if the (contested) EF-TB is violated when student $i$ is allocated $\vc{a_{i\ell}}$ and student $i'$ is allocated $\vc{a_{i'\ell'}}$, i.e. according to Definition~\ref{def:ef-tb} and~\ref{def:strong-eftb}, then we prevent simultaneously allocating $\vc{a_{i\ell}}$ to student $i$ and allocating $\vc{a_{i'\ell'}}$ to student $i'$.

%and any $j\in [k_i], j'\in [k_{i'}]$, if student $i$ envies the $j'$-th bundle of student $i'$ when getting his/her $j$-th bundle, i.e. $u_i(\vc{a_{ij}}) < u_i(\vc{a_{i'j'}})$, we will add the constraint
\begin{align}
  \label{eqn:eftb-check}
  x_{i\ell} + x_{i'\ell'} \leq 1 \text{ if }
  \begin{cases}
   \exists S\subseteq \vc{a_{i'\ell'}}, u_i(\vc{a_{i\ell}})< u_i(S) \;\;\; \text{(for classic EF-TB)}\\
   \exists S\subseteq \vc{a_{i'\ell'}}\cup \{k:p_k=0\}, u_i(\vc{a_{i\ell}})< u_i(S) \;\;\; \text{(for contested EF-TB)}
  \end{cases}.
\end{align}
% \begin{align*}
%     ~.
% \end{align*}

\begin{remark}
  Solving integer programs is NP-hard in general, %Step~\ref{step-np} in Algorithm~\ref{alg:tatonnement} solves an
  but we can solve~\eqref{eqn:budget-perturb-ilp} quite fast in practice with modern SAT solvers. 
\end{remark}

\begin{fact}
Fix parameters $\delta,\epsilon > 0$ and $t \in \{0,1,2\}$.
For a market $(\vc{u},\vc{c})$ and initial budgets $\vc{b_0}$, if Algorithm~\ref{alg:tatonnement} terminates on input $\vc{u},\vc{c},\vc{b_0}$, its output $(\vc{p}^*,\vc{b}^*)$ forms a $\Big(0,\frac{\max_{i\in [n]} (b_0)_i+\eps}{\min_{i\in [n]} (b_0)_i-\eps}\Big)$-CEEI and the final allocation $\vc{a}(\vc{u},\vc{p}^*,\vc{b}^*)$ is EF-TB with respect to $\vc{b_0}$ if $t=1$, or contested EF-TB with respect to $\vc{b_0}$ and $\vc{p}$ if $t=2$.
\end{fact}

%\subsection{Speedups}\label{sub:speedups}

\subsubsection{Speeding up the search of budget perturbations satisfying EF-TB constraints} \label{sub:fast-EF-TB}
The simple implementation of searching EF-TB constraints is that we enumerate all pairs of students $i,i'$ with their possible demands $\vc{a_{i\ell}},\vc{a_{i'\ell'}}$ and check~\eqref{eqn:eftb-check} for all the EF-TB constraints. However, on larger instances this process is very slow because we need to consider  $\binom{n}{2}$ pairs of students, each student may receive one of $k_i$ bundles, and finally for each pair of possible bundles  checking~\eqref{eqn:eftb-check} requires to solve student $i$'s optimal bundle out of a subset of courses $\vc{\tilde{a}_{i'\ell'}}$ ($\vc{a_{i'\ell'}}$ for classic EF-TB and $\vc{a_{i'\ell'}}\cup \{k:p_k=0\}$ for contested EF-TB). For an economy with more than $3000$ students, it requires to solve for more than $4.5$ million such optimal bundles. Repeating that on each iteration is quite slow! %\Aviad{I'm removing the comment about 10 minutes per iteration because that's very hardware-specific} 
For this issue, we consider two optimizations:
\begin{description}
  \item[Two simple sufficient conditions for no-envy:] Fixing $i,i',\ell,\ell'$, we shall use two simple sufficient conditions for proving that allocating $\vc{a_{i\ell}}$ and $\vc{a_{i'\ell'}}$ to $i,i'$ does not violate the (contested) EF-TB constraints. Because these two conditions are easy to verify, we can reduce the number of times we check~\eqref{eqn:eftb-check} a lot and thus save a lot of time. The first condition comes from the reporting language: even if the entire super-bundle $\vc{\tilde{a}_{i'\ell'}}$ is invalid for student $i$, its utility from $i$'s perspective is no greater than that for $\vc{a_{i\ell}}$, i.e.,
  $$B\cdot req_i(\vc{\tilde{a}_{i'\ell'}}) + \vc{w_i}\cdot \vc{\tilde{a}_{i'\ell'}}\leq B\cdot req_i(\vc{a_{i\ell}}) + \vc{w_i}\cdot  \vc{a_{i\ell}}~.$$ 
  The second condition is from the fact that $\vc{a_{i\ell}}$ is the $i$'s optimal allocation under budget $\overline{b_{i\ell}}$. If the total price of bundle $\vc{\tilde{a}_{i'\ell'}}$ is upper bounded by $\overline{b_{i\ell}}$, $\forall S\subseteq \vc{\tilde{a}_{i'\ell'}}, u_i(S) \leq u_i(\vc{a_{i\ell}})$. 
  \item[Memorize envious pairs of students:] Students with very different preferences are likely to never envy each other throughout the run of the algorithm. We take advantage of this idea by only enumerating all pairs of students every 10 iterations, and in the other 9 iterations we  only consider pairs of students whose envy constraint was tight in a past iteration. In particular, when the optimal budget perturbation results in a zero-error solution (under a partial enumeration of possible envies), we force the algorithm to recompute the iteration by enumerating all pairs of students. Note that this implementation cannot guarantee that the allocation computed in each iteration satisfies (contested) EF-TB. However, it guarantee the final allocation satisfies (contested) EF-TB. 
\end{description}
%With these optimizations, on an instance with more than 3000 students, our code can find the EF-TB constraints in about 2 seconds per iteration on average. 
On an instance with approximately 3000 students, these two optimizations speed up our time-per-iteration (amortized including iterations where we check all pairs) by a factor of about 300.   
%\Ruiquan{TODO: two optimizations (memorize the pairs that envies happen \& eliminate )}

\subsection{Shortcuts in price space}\label{sub:Shortcuts}

\paragraph{Warm starts.} In the preliminary experiments, we observe the following phenomena when using different step sizes and proper budget perturbation. 
\begin{itemize}
\item When the step size $\delta$ is small compared to $\epsilon$, the algorithm can converge to a zero-error solution. However, for courses that are consistently slightly  over-demanded, their prices increase slowly from $0$ to their final prices. With these courses, the algorithm needs almost $1/\delta$ or even more steps to converge. 
\item On the other hand, when the step size $\delta$ is large, the number of possible budget-demand pairs for each student may not be enough to help the budget perturbation significantly improve the clearing error. However, even if we set $\epsilon$ to $0$ (i.e. we only use discrete tatonnement), the prices found by the algorithm can be quite close to good regions, where prices found with smaller $\delta$ lie, and the algorithm only takes much less time to reach such regions because of the larger step size.
\end{itemize}
Motivated by these observations, we shall combine discrete tatonnement and our algorithm together to improve the speed. We shall first run discrete tatonnement with a larger step size $\delta_0$ and with $(1+\beta)/\delta_0$ steps, and then turn to our algorithm with a smaller step. In the first warm-start phase we also save time by not computing optimal budget perturbations. %\Aviad{@Ruiquan, do we use budget perturbations in the warm start?} \Ruiquan{No.}

%\Ruiquan{Qianfan uses the multi-step warm start process (no EF-TB -> simple EF-TB -> strong EF-TB) in the manipulability part. However, it seems not to help much in the performance part because the envies are quite sparse compared to the number of students. Should we write up this process?}
 
\paragraph{Merge Equivalent Steps.} 
% In preliminary experiments, we observe that there can be a lot of consecutive equivalent steps that share the same updates (i.e. the optimal excess demand $\vc{\tilde{z}}(\vc{u},\vc{c},\vc{p},\vc{b})$), especially when the algorithm has already found some solution with very small clearing error. For example, when there is only one course with capacity 1 and two students want the course and the price is less than $1$, our algorithm always increases the price by $\delta$. In this case, our algorithm needs at least $1/\delta$ such steps to move the price from $0$ to $1$. Intuitively, these steps waste a lot of time because we can replace it with a single step with step size $1$. 

As discussed above, when we use smaller step sizes, it may take the algorithm many iterations to update some prices. 
Fortunately, it turns out that for many of those iterations the set of possible demands remains constant across all students for many consecutive iterations. 
Whenever this is the case, we can save time by binary searching for the next iteration where the excess demand changes. %\Ruiquan{I am not sure whether the optimal excess demand of the budget perturbation changes is binary searchable. What we have implemented is searching for the next iteration where some student changes his set of possible demands under budget $1\pm \eps$.} \Aviad{So if I understand correctly the {\em set} of feasible student demands is constant for a stretch of several iterations?}
A further more clever optimization considers stretches where the demand sets alternate between only a few possible vectors; again we can binary search for the number of iterations that we need to take to reach a new demand vector. % \Aviad{Since there are two demand vectors here, in which direction do we go?} \Ruiquan{We should go in the direction of the sum of the two demand vectors. PS: we also implement this trick when the demand alternates for 3 or more vectors.} \Aviad{Maybe you could write a paragraph to elaborate on that?}

\section{Computational performance of our algorithm}
\label{sec:perform}
In this section, we discuss the computational performance of our algorithm, in particular in comparison to the previous state of the art.
In Subsection~\ref{sub:compare-baseline}, we 
% begin by describing the previous state-of-the-art algorithm (Appendix~\ref{app:baseline}), then 
describe our experiments comparing the algorithms. 
Then, in Subsection~\ref{sub:compare-opt-algo} we focus on how the improvements described in Section~\ref{sub:Shortcuts} compare to our basic algorithm.

\subsection{Comparing with the benchmark}
\label{sub:compare-baseline}

In this subsection we compare our algorithm with the benchmark algorithm. As we will soon see (Figure~\ref{fig:progress-of-algos-error}), our algorithm is much faster. To understand why, we also consider two intermediate algorithms. Overall, the algorithms we compare are:

\begin{description}
    \item[Benchmark] The (previous) state-of-the-art commercial algorithm; see description in  Appendix~\ref{app:baseline}.
    \item[Vanilla] Vanilla tatonnement, aka without optimizations such as tabu search, individual price adjustments, or optimal budget perturbations, that are used in other variants that we consider (for pseudocode, see  Algorithm~\ref{alg:tatonnement} with $\epsilon = 0$). 
    \item[Basic] Our basic algorithm, aka the algorithm descirbed in Subsection~\ref{sub:basic}, which adds optimal budget perturbations to tatonnement, but without further optimizations described in Subsection~\ref{sub:Shortcuts}.
    \item[Fast] Our final algorithm, including optimizations from Subsection~\ref{sub:Shortcuts}.
\end{description}

% Number of courses: 173
% Number of students: 435
% Number of courses: 241
% Number of students: 445
% Number of courses: 181
% Number of students: 468
% Number of courses: 156
% Number of students: 473

\paragraph{Choice of parameters.}
All the algorithms perturb the students' budgets: Benchmark and Vanilla use only random perturbations, while Basic and Fast use a mix of random and optimal budget perturbations. 
Larger budget perturbations tend to make the algorithmic task of finding an (approximate) equilibrium easier: At one extreme, the existence proof holds even with infinitesimal budget perturbations; at the other extreme, with infinite budget perturbations the mechanism reduces to Random Serial Dictatorship which is computationally trivial.

For the sake of a fair comparison of the algorithms, we ensure the algorithms have the same magnitude of total budget perturbations: 
\begin{itemize}
    \item For Benchmark and Vanilla we draw the budgets uniformly i.i.d.\ from $[1,1+\beta]$.
    \item For Basic and Fast, we draw the {\em base budgets} uniformly i.i.d.\ from $[1+\beta/4,1+3\beta/4]$, and allow further optimal budget perturbation of magnitude $\epsilon = \beta/4$.
\end{itemize}
%In particular, our experiments focus on $\beta = 0.04$ and $\beta=0.02$. We set the step $\delta$ to be $0.002$ for $\beta=0.04$, while we set $\delta$ to be $0.001$ for $\beta=0.02$. 
In particular, we consider $\beta=0.04$ here. 
We use a large step size of $0.005$ for the warm start of Fast. 
For the remaining parameters, we replicate that of~\cite{BudishCKO17} for Benchmark and use the same step size of $\delta=0.002$ for all the algorithms (except in the warm start).

\paragraph{Instances.} For the computational experiments, we use the largest instances available to us:

% Due to the application to the real course allocation systems, some instances may involve unequal base budgets. 
% \Aviad{Ruiquan, for now, please include as much information as possible, in particular instance sizes, but also school names and years. We can check in later with Abe et al if that's OK...}
% \Ruiquan{We test on instances from the following schools:
\begin{itemize}
\item {\bf Law} - A law school with about 500 students and 125 classes/sections % Toronto Law
\item {\bf Ivy Large} - A large Ivy-league business school with several thousand students and several hundred classes/sections % Wharton
\item {\bf Ivy Huge} - A large Ivy-league business school with several thousand students and around a thousand classes/sections and lots of challenging constraints % Columbia
\end{itemize}

A typical student's schedule in those instances has around 5 courses, although some schools bid in the fall for the entire year (so around 10 courses total).

\subsubsection*{Findings}

\begin{figure}[t]
\centering
% \includegraphics[width=.49\textwidth]{img/CMP_compressed_big_time.pdf}
% \,
% \includegraphics[width=.49\textwidth]{img/CMP_compressed_big_iter.pdf}
% \\
\includegraphics[width=.48\textwidth]{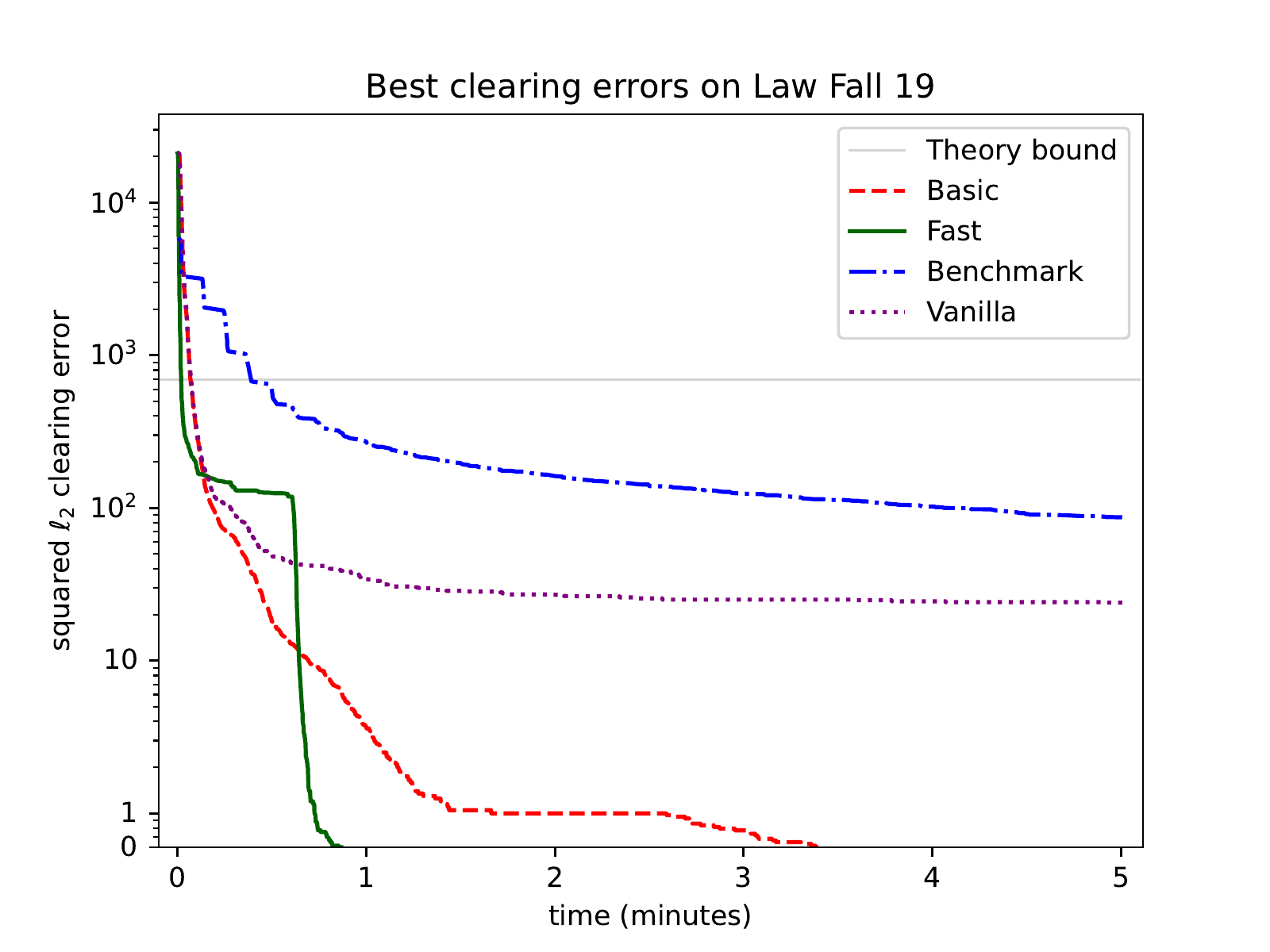}
% \,
\includegraphics[width=.48\textwidth]{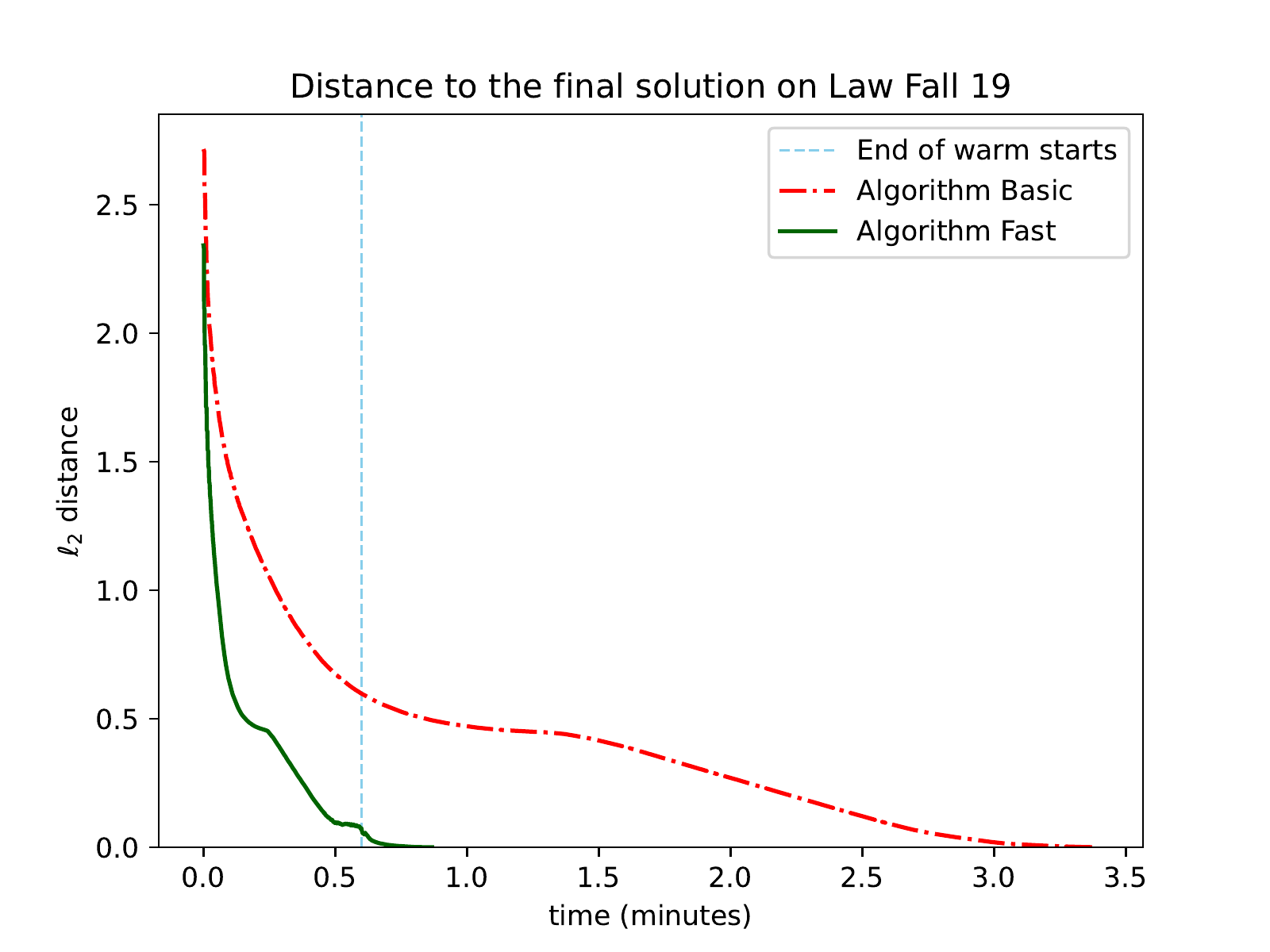}
% \includegraphics[width=.45\textwidth]{img_errors/large-robust-errors-wrt-time.pdf}
% \\
% \includegraphics[width=.45\textwidth]{img_errors/ordata-robust-errors-wrt-time.pdf}
% \,
% \includegraphics[width=.45\textwidth]{img_errors/giant-robust-errors-wrt-time.pdf}
\caption{For each of the algorithms tested and time $t$, we plot the best clearing error obtained by the algorithm up to time $t$ (left). For Basic and Fast, we plot the distance between the found solution and the final solution with respect to the running time (right).}
\label{fig:progress-of-algos-error}
\label{fig:progress-of-algos-distance}
\end{figure}

Figure~\ref{fig:progress-of-algos-error} presents the average%
\footnote{
This is an average over different draws of students' initial budgets; we use 20 runs for Basic and Fast, and 10 runs for the slower Benchmark and Vanilla. %(For aesthetic purposes, because the $y$-axis is drawn on log scale, When the average $y$ is below $1$, we replace it by $0.7+0.3x$.)
} 
best clearing error found by the four algorithms with respect to time.%  \Aviad{Average over how many attempts?}. \Ruiquan{Our algorithms for 20 times and the baseline & tatonnement for 10 times.}

\paragraph{Observation 1: optimal budget perturbations help, significantly.}
From the plot, we can see that the two algorithms with budget perturbation (i.e.~Basic and Fast) significantly outperform those without budget perturbation (i.e.~Benchmark and Vanilla) on clearing errors --- when our basic algorithm terminates with a $(0, \beta)$-CEEI, the average best clearing errors found by Benchmark and Vanilla are still larger than 30. Therefore, we believe budget perturbation is the most important ingredient introduced in our algorithms for clearing the market. 

\paragraph{Observation 2: individual price adjustments do not help.}
In Figure~\ref{fig:progress-of-algos-error}~and~\ref{fig:progress-of-algos-error-app} it can be observed that Vanilla obtains lower clearing error than Benchmark on all the instances we tested. More precisely if we run both algorithms long enough then Benchmark does eventually catch up (Figure~\ref{fig:progress-of-algos-error-4h}), but only at a time scale much larger than our algorithm needs to converge to zero clearing error.

The observation that individual price adjustments slow the algorithm is very surprising since it stands in contrast to the findings of~\cite{OthmanSB10}. It is even more puzzling because in each iteration Algorithm~\ref{alg:tabu-search} chooses the (myopically) better of updating all the prices or one, so it seems intuitive that individual price adjustments can only help. %\Aviad{@Ruiquan - I rewrote this below. Please take a look and lmk if this is a faithful description of observations.} \Ruiquan{I think so!}
One simple reason is that the time-per-iteration of Benchmark is significantly slower compared to Vanilla\footnote{We note that in the experiments we actually measure a proprietary variant of the benchmark algorithm whose time-per-iteration has been heavily optimized. We also note that when we tried to combine individual  price adjustments with optimized budget perturbations, the computational overhead of individual price adjustments was even worse because we had to resolve for optimal budget perturbation for each individual price adjustment.}. (In~\cite{OthmanSB10} this issue did not come up because  both algorithms were tested in terms of number of iterations.) 

Another interesting piece of this puzzle is that~\cite{OthmanSB10} thought of tatonnement as computer scientists often do: a direction for a local search algorithm with the objective of minimizing the clearing error (indeed, in some utility models it exactly corresponds to a (sub)-gradient of the clearing error, e.g.~\cite{KC82,CheungCD20,LW20}). When viewed in this way, it may sometimes get stuck in local minima. However, inspired by economics, we view tatonnement as a distributed process where we modify the price of each course simultaneously without worrying about the global clearing error. When used in this way, it may sometimes locally {\em increase} the clearing error; but our experiments suggest that it can effectively escape those local minima, leading to better solutions faster. Indeed the clearing error does not monotonically decrease throughout the run of our algorithm.

\paragraph{Observation 3: Our optimized algorithm initially has larger clearing errors than our basic algorithm.} In the plot, it is easy to spot the sharp drop of clearing error for Fast --- this is exactly the time when Fast switches from the warm start to the second phase. It may appear that the time we spend on the warm start is too long, but this is in fact not the case, as we discuss in Subsection~\ref{sub:compare-opt-algo}.

\subsection{Our optimized algorithm }
\label{sub:compare-opt-algo}

\begin{comment}
\begin{figure}[t]
    % \includegraphics[width=.32\textwidth]{img_warmstart/s3-toronto-law-fall-19-normal-distances-wrt-time.pdf}
    % \,
    % \includegraphics[width=.32\textwidth]{img_warmstart/large-normal-distances-wrt-time.pdf}
    % \,
    % \includegraphics[width=.32\textwidth]{img_warmstart/giant-normal-distances-wrt-time.pdf}
    % \\
    \centering
    \includegraphics[width=.5\textwidth]{img_warmstart/s3-toronto-law-fall-19-robust-distances-wrt-time.pdf}
    % \,
    % \includegraphics[width=.45\textwidth]{img_warmstart/large-robust-distances-wrt-time.pdf}
    % \\
    % \includegraphics[width=.45\textwidth]{img_warmstart/ordata-robust-distances-wrt-time.pdf}
    % \,
    % \includegraphics[width=.45\textwidth]{img_warmstart/giant-robust-distances-wrt-time.pdf}
    % \,
    % \\
    % \includegraphics[width=.49\textwidth]{img/ordata-distances-wrt-time.pdf}
    % \,
    % \includegraphics[width=.49\textwidth]{img/ordata-distances-wrt-iter.pdf}
    % \\
    % \includegraphics[width=.49\textwidth]{img/ordata-distances-wrt-time-2.pdf}
    % \,
    % \includegraphics[width=.49\textwidth]{img/ordata-distances-wrt-iter-2.pdf}
    % \\
    % \includegraphics[width=.49\textwidth]{img/giant-distances-wrt-time.pdf}
    % \,
    % \includegraphics[width=.49\textwidth]{img/giant-distances-wrt-iter.pdf}
    \caption{The distance between the found solution and the final solution with respect to the running time.}
    \label{fig:progress-of-algos-distance}
\end{figure}
\end{comment}

As mentioned in Observation 3 above, Fast initially makes slow progress in terms of market clearing error during the warm start compared to Basic. We argue that it still makes important progress towards the eventual equilibrium even if we don't see that in the clearing error.
First, we note that although we did not carefully optimize the cutoff of the warm start (we heuristically set it to approximately $1/\delta$, where $\delta$ is the step size),  Fast seems to work really well in practice!
%we experimented with different cutoffs for the warm start, and the cutoff we chose seems to be approximately optimal towards minimizing the total time to zero clearing error. \Aviad{@Ruiquan, is that true? You had some heuristic to determine this cutoff, like $1/\delta$ steps? Something like that?} \Ruiquan{I just observed that it is approximately enough for the course prices to move to a ``good'' region, so sometimes there may be too many iterations. The choice is something like $1/\delta$ for equal budgets.}

More interestingly, we can measure the progress towards the eventual equilibrium in terms of distance in price space (instead of current clearing error). In Figure~\ref{fig:progress-of-algos-distance}~and~\ref{fig:progress-of-algos-distance-app}, we show that the optimized algorithm approaches an ultimate equilibrium much faster when measured in price space distance. We also observe that at the time that the algorithm switches phases, we're already quite close in price space, so a more refined (small step size) second phase is appropriate. (Of course we unfortunately only know the distance to eventual equilibrium in hindsight, otherwise this could have made for a great heuristic approach to knowing when to switch phases!)

\section{Examples of successful manipulations}
\label{sec:anecdotes}
How can strategic students manipulate SP-L  mechanisms in realistic instances? 
To really understand the nature of manipulations found by the algorithm, in this section we report insights from our qualitative analysis that zooms in on specific manipulations, one for each variant of our algorithm. Our representation is over simplified with made up student names and courses --- but they're all based on manipulations discovered by our manipulation-finding algorithm on almost-real data (see Remark~\ref{remark:5-additive}). Of course, each case study may not be representative of all possible manipulations. However, we find this methodology quite helpful for gaining intuition. In particular, we were able to use the case study for classic EF-TB to extract the simple manipulation described in the introduction, and propose the contested EF-TB criterion in response. 

All the examples are described in what can be informally thought of as an ``almost-large-market'': from the perspective of an individual student, course prices are approximately set by other students, but a student's reported preferences can nudge them infinitesimally  towards a market clearing equilibrium. 

\begin{remark}[5-additive utilities]\label{remark:5-additive}
Unfortunately, the original preferences in the true instances are incredibly complicated, mostly due to various constraints imposed by the schools (e.g.~avoiding courses with conflicting meeting times, meeting minimum unit requirements, etc). 
To keep the case study analyses tractable, we repeat the manipulation finding experiments, but on modified preferences that ignore all those constraints and simply assume that the students utility is {\em 5-additive}, i.e.~each student wants the schedule of 5 courses that maximizes the total weight; here we use the original course weights reported by the students, but ignore all conflicts and requirements. Formally, for every student $i\in [n]$, the new utility function can be described by tuple $(\vc{w}',valid',req')$ where $\vc{w}'=(w_j)_{j\in [m]}$, and for every bundle $\vc{x}\in 2^{[m]}$,
    \begin{align*}
        valid'(\vc{x})&=\begin{cases}1&\|\vc{x}\|_1\le 5\\0&\|\vc{x}\|_1>5\end{cases}\\
        req'(\vc{x})&=0
    \end{align*}
    \end{remark}
    
    \subsection{A simple manipulation without EF-TB}
    
    \begin{example}[Manipulation without EF-TB constraints] \label{ex:no-EFTB} \hfill
    
    Alice and Bob both want only the last seat of CS161. Because of other students' demand,  ECON101 is full but has a low price.
    With true reporting CS161 could go to either Alice or Bob, depending on the random initial budgets. 
    
    Bob can manipulate his preferences to report that he wants ECON101 as his second course (aka Bob's manipulated preferences are:
    $$
    \{\text{CS161},\text{ECON101}\} \succcurlyeq \{\text{CS161}\} \succcurlyeq \{\text{ECON101}\}.
    $$
    Under the manipulated preferences, since ECON101 is cheap, the only way it will not be allocated to Bob is if Bob exhausts his budget on CS161. Even if Alice's initial budget is higher, the optimal unconstrained budget perturbation will make sure that Bob's final budget is higher than Alice's, in which case he gets CS161 and Alice gets nothing --- an equilibrium.
    \end{example}
    
    For this manipulation to work, Bob had to know that at equilibrium prices ECON101 is already exactly filled by other students  --- a knowledge he is unlikely to have in realistic bidding. Indeed, suppose that demand for ECON101 was higher this semester, the algorithm raised its price, and now it is missing exactly one student: in this case the optimal budget perturbation would have to ensure that Bob does get into ECON101, which can be achieved by perturbing budgets {\em against} Bob and letting Alice grab the last seat in CS161.
    
    However, this manipulation is fairly robust if the price of ECON101 is always very low (``ECON101 tends to have a low price'' is a general statistic that a student could plausibly learn from historical bids). In this case, even if ECON101 is missing a student, Bob's budget may be larger than Alice's by a sufficient margin to afford both CS161 and ECON101. So when ECON101 is undersubscribed, the budget perturbation could go either way\footnote{Alice would have a slight advantage due to our particular tie-breaking rule.}, but it {\em always} goes in favor of the strategic Bob when ECON101 is oversubscribed.
    
    \subsection{A simple manipulation with classic EF-TB}
    
\begin{example}[Manipulation with classic EF-TB constraints]\label{ex:regular-manipulation} \hfill

Alice and Bob both want the last seat in the popular CS161 course. But Alice is even more excited about taking independent research units with her advisor, of which there is unlimited supply, so the price is always zero (aka this is an {\em uncontested} course); she would like to take both. Bob's second choice is  ECON101; because of other students' demand,  ECON101 is full but has a low price.

With true preferences, since ECON101 is cheap, the only way it will not be allocated to Bob is if Bob exhausts his budget on CS161. Thus even if Alice ranks higher, the optimal budget perturbation sets her budget lower than Bob. In this case Alice always gets independent study (only), and Bob always gets CS161 (only) --- an equilibrium.

If Alice misreports her preferences to rank independent research units lower than CS161, she would envy Bob whenever he gets the last seat to CS161. Whenever her initial budget is higher, this prevents the optimal budget perturbation from driving it below Bob's, increasing her chances of getting the last seat in CS161. 
\end{example}

Note that ranking the uncontested course (independent research units) lower {\em never hurts} Alice. 
So this manipulation is profitable in expectation for Alice even if she only has very noisy information about her rank and other students' demand.

Interestingly, this manipulation works {\em because} of the EF-TB constraints that we introduce to prevent manipulations! %In Section~\ref{sec:case-study} we also describe a simple manipulation that works when EF-TB constraints are not enforced. 
    
    \subsection{A simple manipulation with contested EF-TB}
    
    We now discuss a simple manipulation that our algorithm discovered even with the contested EF-TB; while some profitable manipulations may exist in practice, as we show in Section~\ref{sec:manip} they're extremely rare. 
    
    \begin{example}[Manipulation with contested EF-TB constraints]\label{ex:contested} \hfill
    
    Many students, including Bob, like to take CS161 and ECON101 together, but they rank ECON101 over CS161. Alice already took ECON101 last semester, so she only wants the last seat in CS161. Because of other students' demand,  ECON101 is full but has a very low price.

    With true preferences, if Bob's budget is higher than Alice's by a margin greater than the price of ECON101, he could afford both CS161 and ECON101, leaving Alice with nothing. 
    
    Alice could manipulate her preferences to report that she wants ECON101 as her second course. 
    Bob always gets ECON101, because it's cheap and it's his top priority.  Whenever Alice doesn't get CS161 she has to get ECON101 (because its price is cheap); if both Bob and Alice get ECON101, the course becomes oversubscribed, which would cause a clearing error. Therefore the optimal budget perturbation would make sure Bob's budget is low enough compared to Alice's that he can't afford both courses: he will get ECON101, and Alice will get CS161 - an equilibrium.
    \end{example}
    
    As with Example~\ref{ex:no-EFTB}, this manipulation {\em does} pose some risk --- if ECON101 is undersubscribed, the optimal budget perturbation may reduce Alice's budget below the price of CS161 so that she has to take ECON101. However, if ECON101 is very cheap, there's always also the small perturbation that increases Alice's budget so that she can afford both ECON101 and CS161 (while Bob only affords ECON101). So, because of the asymmetry in the prices of CS161 and ECON101, if ECON101 is oversubscribed, this manipulation can increase Alice's chances of getting into CS161, but if ECON101 is undersubscribed, Alice's chances aren't hurt by much. 
    
        \subsection{Can A-CEEI with random budget perturbation be manipulated?}
        
        All the manipulations that we found seemed tied to our optimal budget perturbations procedure. So it is natural to ask whether the original A-CEEI%
        \footnote{For fair comparison, note that the previous state-of-the-art practical implementation augmented the original A-CEEI mechanism with two stages that had other incentive issues (see Remark~\ref{remark:zero-error}).} without optimal budget perturbations can also be manipulated
        
        We speculatively conjecture in practice that manipulations similar to those described in Example~\ref{ex:contested} can also be profitable without optimal budget perturbations: the simplest way to think about this example is that Alice adding ECON101 to her demand should (slightly) increase the price of ECON101; this makes it less likely for Bob to be able to afford both ECON101 and CS161; this in turn makes it more likely that Alice can get a seat in CS161.
        
        While it seems plausible that such manipulations are profitable in practice, note that for our algorithm with contested EF-TB, our numerical analysis in Section~\ref{sec:manip} suggests that they are extremely rare. 
        Either way, at this point for A-CEEI without optimal budget perturbations we can only speculate: we could not run our manipulation-finding algorithm with the original A-CEEI algorithm that only uses random budget perturbation because this algorithm is too slow. (The manipulation-finding algorithm needs to make many calls to the A-CEEI algorithm to evaluate different possible deviations for every student.)

\section{Quantitative analysis of manipulability}
\label{sec:manip}

In this section we ask whether in practice strategic students bidding in A-CEEI have an incentive to deviate from truth-telling. 
To address this question, we model the student's process for choosing her bids with an optimization algorithm that can try different bid manipulations and test whether they improve the student's utility.

\subsection{The manipulation-finding algorithm}

It is still intractable to consider all possible bid-manipulations. This is true both for a student optimizing their bid, and for our algorithms in our experiments. Instead, we restrict attention to a simple hill-climbing algorithm that iteratively looks for a course whose bid-manipulation would increase the utility, in expectation over uncertainty; see Algorithm~\ref{alg:find-manipulation} for details. To validate our approach, we benchmark our manipulation-finding algorithm on a course allocation mechanism that is known to be manipulable~\cite{BC12}, which is used at Harvard Business School (we henceforth refer to this mechanism as HBS).  

To simplify notations, we explicitly define the notation of randomized mechanism for course allocation problem.
\begin{definition}[Randomized mechanism]
\label{def:randomized-mechanism}
A randomized mechanism $\calM$ for course allocation problem can be characterized by a function $f_\calM: (\vc{u},\vc{c},\vc{r}) \mapsto (\vc{a}_i)_{i=1}^n$, which takes the input of the course allocation market $(\vc{u},\vc{c})$ and randomness $\vc{r}$, and outputs an allocation $(\vc{a}_i)_{i=1}^n=(\vc{M}_i(\vc{u},\vc{c},\vc{r}))_{i=1}^n$ where $\vc{M}_i(\vc{u},\vc{c},\vc{r}) \subseteq [m]$ denotes the bundle that the mechanism allocated to student $i$.

For example, HBS is a randomized mechanism, where the randomness $\vc{r}$ is used to determine the order of students in the random serial dictatorship process.
Our A-CEEI algorithm (Algorithm~\ref{alg:tatonnement}) with fixed parameters can also be seen as a randomized mechanism, since an allocation can be uniquely determined based on its output (prices and budgets), and the base budgets are determined by randomness $\vc{r}$.
\end{definition}

\begin{remark}
    Note that in Definition~\ref{def:randomized-mechanism}, we do not require the allocation generated by the mechanism to be feasible (i.e., some courses might be oversubscribed). That is because no A-CEEI algorithm can guarantee to always output a price with zero clearing error. Nevertheless, our algorithm was observed to always obtain a feasible solution in all instances we encountered, as mentioned before.
\end{remark}

In our experiments on manipulability of a specific mechanism, we iteratively run (a variant of) Algorithm~\ref{alg:find-manipulation} multiple times with respect to different parameters $H=\{\eta_1,\eta_2,...\}$ for every student, under the different uncertainty models described later in Definition~\ref{def:manipulation}. 

\begin{algorithm}[t]
\caption{Find a profitable manipulation for a student}
\label{alg:find-manipulation}
\begin{algorithmic}
\State \textbf{Inputs:} a randomized mechanism $\calM$, student $i$ and its utility function $u_i$, (previous best manipulation $v_0$), \\ the criteria for profitable manipulation (resampled randomness $(\vc{u}_{-i},\vc{c},\calR)$ or population $(\calU_{-i},\vc{c},\calR)$)
\State \textbf{Outputs:} a profitable manipulation $u'_i$
\State \textbf{Parameters:} a local update coefficient $\eta$
\State \textbf{Algorithm:}
\begin{enumerate}
  \item Let $v_0 \leftarrow u$ (or the best manipulation found in previous iterations with different $\eta$).
  \item Denote the description for $v_0$ by $(\vc{w},valid,req)$.
  \item Try to increase or decrease the weight $w_j$ for each course $j$ in $v_0$ to obtain new misreports $V = \{v_{j,\pm 1}\}_{j\in [m]}$. Each $v_{j,k}$ can be described by $(\vc{w}',valid,req)$ where
  $w'_{j'}=\begin{cases}w_j&j' \ne j\\ \eta^k w_j&j' = j\end{cases}$ for $j'\in[m]$, $k \in \{\pm 1\}$.
  \item Let $v^*=\begin{cases}%\argmax_{v \in V_1 \cup \{v_0\}} u_{i}([\vc{A}_i([v_j,\vc{u}_{-i}],\vc{c},\vc{b_0})) & \text{complete information}, \\ 
  \argmax_{v \in V \cup \{v_0\}} \E_{\vc{r}\sim \calR}[u_i(\vc{M}_i([v_j,\vc{u}_{-i}],\vc{c},\vc{r}))] & \text{resampled randomness}, \\ 
  \argmax_{v \in V \cup \{v_0\}} \E_{\vc{u}_{-i} \sim \calU_{-i},\vc{r}\sim \calR}[u_i(\vc{M}_i([v_j,\vc{u}_{-i}],\vc{c},\vc{r}))] & \text{resampled population}.\end{cases}$
  \item If $v^* = v_0$, terminate with $v_0$ as the best manipulation found when $v_0 \ne u$, otherwise return failed.
  \item Otherwise, update $v_0 \leftarrow v^*$ and go back to step 2.
\end{enumerate}
\end{algorithmic}
\end{algorithm}

\paragraph{Practical optimizations}
%\begin{remark}
Since computing the exact expected value in the step 4 is generally intractable, we have to use samples to estimate it. 
For computational efficiency, we only use 5 samples for each estimation in following experiments, which is a fairly small number. As a result, some profitable manipulations might be missed, and some non-profitable manipulations can be reported incorrectly.
Nevertheless, we will use a larger number of samples to test the statistical significance of those manipulations found by Algorithm~\ref{alg:find-manipulation}.

We also use another optimization to reduce the number of times we have to solve for an A-CEEI, we search for manipulations in parallel: in each run of Algorithm~\ref{alg:find-manipulation} we actually consider a subset of several students and who are trying out their deviation on the same instance (students subsets are shuffled in each iteration to ensure that the signal corresponds to the deviation by the student and not to particular deviations by their peers).

% \Aviad{@Qianfan discuss learning rate}

\paragraph{Handling false positive manipulations}
For any manipulation that is profitable on average in the exploration phase (which, as mentioned above is very noisy), we run it for $33$ iterations, and then $67$ more if it is still profitable. At this point we eliminated manipulations whose profitability did not meet $.05$ p-value statistical significance. 

The original experiments had many false positives, and we expect as many as $5\%$ of them would to survive the p-value test. So for all manipulations that survived the first p-value test, we run $100$ additional iterations, and report manipulations that still have p-value $.05$. 

%\Qianfan{TBD- mention the batch evaluation trick which is used in both finding deviations and verifying deviations, where we use a fixed batch for finding and a random batch for validation.}
%\end{remark}

\subsection{Modeling students' uncertainty}
%Our first observation \Aviad{TBD: statistics? graphs?} is that for many students we can find some improving deviation from truthful bidding. However, computing these improving deviations requires perfect knowledge of both the market (aka course capacities and other students' preferences), as well as the internal randomness of our algorithm. 
In practice, while students often know the course capacities, they tend to have limited, aggregate information about historical bids; furthermore it is impossible for the students to know the algorithm's randomness at the time of bidding. We therefore study the students' decision process under two different models of uncertainty: {\em Resampled randomness} is the most conservative model of students' uncertainty --- even if they have perfect knowledge of their peers' preferences, they never know the randomness of the algorithm (in particular, the random budget perturbation) at the time of bidding. {\em Resampled population} adds the students' uncertainty about their peers' preferences. See details of both models in Definition~\ref{def:manipulation}.

\begin{definition}[Profitable manipulation]\label{def:manipulation}
For a randomized mechanism $\mathcal{M}$, student $i$'s original utility function $u_i$ and some manipulation $u_i'$, we say the manipulation from $u_i$ to $u_i'$ is {\em profitable}
\begin{itemize}
    % \item \Aviad{@qianfan} under {\em complete information} $(\vc{u}_{-i},\vc{c},\vc{b_0})$: profitable under the market $([u_i,\vc{u}_{-i}],\vc{c})$ and the initial budget $\vc{b_0}$, i.e.,
    % $$u_i(\vc{A}_i([u_i',\vc{u}_{-i}],\vc{c},\vc{b_0})) > u_i(\vc{A}_i([u_i,\vc{u}_{-i}],\vc{c},\vc{b_0})).$$
    \item under {\em resampled randomness} $(\vc{u}_{-i},\vc{c},\calR)$: profitable in expectation under the market $([u_i,\vc{u}_{-i}],\vc{c})$ and randomness $\vc{r}$ where $r \sim \calR$, i.e.,
     $$\E_{\vc{r}\sim \calR}[u_i(\vc{M}_i([u_i',\vc{u}_{-i}],\vc{c},\vc{r}))] > \E_{\vc{r}\sim\calR}[u_i(\vc{M}_i([u_i,\vc{u}_{-i}],\vc{c},\vc{r}))].$$
     %We let distribution $\calB_0=U[1-\beta',1+\beta']^n$ for some small constant $0<\beta'<1$. \Aviad{@Qianfan we need to change to $[1,1+\beta']$... does it have any implications anywhere? Do you report this parameter? Also, why $\beta'$ vs $\beta$?}
    \item under {\em resampled population}
    %\Aviad{could be good to explicitly say how we resample the population?} \Qianfan{Done}
    $(\calU_{-i},\vc{c},\calR)$: profitable in expectation under the market $([u_i,\vc{u}_{-i}],\vc{c})$ and randomness $\vc{r}$ where $\vc{u}_{-i} \sim \calU_{-i}$ and $\vc{r} \sim \calR$, i.e.,
    $$\E_{\vc{u}_{-i} \sim \calU_{-i},\vc{r}\sim \calR}[u_i(\vc{M}_i([u_i',\vc{u}_{-i}],\vc{c},\vc{r}))] > \E_{\vc{u}_{-i} \sim \calU_{-i},\vc{r}\sim\calR}[u_i(\vc{M}_i([u_i,\vc{u}_{-i}],\vc{c},\vc{r}))].$$
    %As an example, based on a known market $(\vc{u},\vc{c})$, for every student $i$ we can let $\calU_{-i}$ be the distribution of drawing samples from $\vc{u}_{-i}$  without replacement. \Aviad{
    Given a known market $(\vc{u},\vc{c})$, for every student $i$ we  let $\calU_{-i}$ be the distribution that re-samples $n-1$ other students, independently and with replacement, from of the true population $\vc{u}_{-i}$ of other students.%} \Qianfan{That should be precise.}
\end{itemize}
\end{definition}

\subsection{Experiment set-up}

\paragraph{Choice of parameters} For our A-CEEI algorithm (Algorithm~\ref{alg:tatonnement}), we always seek for a $(0,\beta)$-CEEI, where $\beta=0.04$. 
Same as before, we draw base budgets $\vc{b}_0$ uniformly i.i.d.~from $[1+\beta/4,1+3\beta/4]$ (i.e., $\vc{b}_0=\vc{r} \sim \calR=\mathcal{U}^n[1+\beta/4,1+3\beta/4]$) and allow further optimal budget perturbation of magnitude $\epsilon=\beta/4$.
We choose different step sizes for different instances to speed up computation. Specifically, we use $\delta=0.0005$ for Ivy Small, $\delta=0.001$ for Biz, and $\delta=0.002$ for Small.

For the manipulation-finding algorithm (Algorithm~\ref{alg:find-manipulation}), we let the set $H$ of local update coefficients be $\{16000, 800, 40, 2\}$.
% When evaluating deviations, we take the average over 5 samples for each deviation in the exploration phase. (This number is quite small but we have to evaluate many candidate deviations and each one requires solving an A-CEEI instance.) For deviations that are successful in the first exploration phase, we further validate them over 100 samples. 

\paragraph{The instances}

Despite various optimizations, our manipulation-finding ultimately  requires solving a very large number of A-CEEI instances, which can require significant computational resources even with our efficient algorithm. Therefore we run this experiment on relatively small instances:
\begin{itemize}
    \item {\bf Small} - A small business school with about 150 students and 50 classes/sections (Fall 2021).
    \item {\bf Ivy Small} - An Ivy-league business school with about 500 students and 50 classes/sections (Fall 2020).
    \item {\bf Biz} - A business school with about 500 students and 50 classes/sections (Fall 2020).
\end{itemize}

%The new instance is \texttt{s3-tuck-fall-2020.in}, where there are 40 courses and 283 students.

% \begin{figure}[H]
%     \centering
%     \includegraphics[scale=0.6]{img_instances/s3-tuck-fall-2020.inall_courses_analysis.pdf}
%     \caption{Weight distribution of~\texttt{s3-tuck-fall-2020.in}}
% \end{figure}

%\Qianfan{Rotman: (do we still need them?)
%Despite various optimizations, our manipulation-finding ultimately  requires solving a very large number of A-CEEI instances, which can require significant computational resources even with our efficient algorithm. Therefore we run this experiment on relatively small instances. 
%The new instance is \texttt{s3-rotman-fall-2020.in} and \texttt{s3-rotman-fall-2021.in}. In the 2020 instance, there are 52 courses and 441 students.
%In the 2021 instance, there are 54 courses and 437 students.
%Among these courses, there are 28 of them that are identical in 2020.}

% \begin{figure}[H]
%     \centering
%     \includegraphics[scale=0.6]{img_instances/s3-rotman-fall-2020.inall_courses_analysis.pdf}
%     \caption{Weight distribution of~\texttt{s3-rotman-fall-2020.in}}
% \end{figure}

% \begin{figure}[H]
%     \centering
%     \includegraphics[scale=0.6]{img_instances/s3-rotman-fall-2021.inall_courses_analysis.pdf}
%     \caption{Weight distribution of~\texttt{s3-rotman-fall-2021.in}}
% \end{figure}

\subsection{Statistical findings}
Our numerical results are summarized in Table~\ref{tab:stat-manipulation}. Our algorithm successfully finds profitable manipulations for the benchmark manipulable HBS mechanism on all instances. On instances Small and Biz  it finds almost no statistically significant profitable manipulations for any variant of our A-CEEI algorithm. For Ivy Small, it finds that about 7\% of the students can gain as much as around 10\% from misreporting their preferences with no EF-TB constraints, and somewhat less with classic EF-TB constraints; however with contested EF-TB profitable manipulations were extremely rare.

%\Aviad{TBD - rewrite this section once we have some more results}

%\subsubsection{Profitable manipulations}

% \Qianfan{Currently running experiments further testing those manipulations with p-value < 0.05. Should be out soon.}

% \begin{enumerate}
%     \item Use (sampled version of) Algorithm~\ref{alg:find-manipulation} to find possible manipulations.
%     \item Verify those candidates to find real profitable manipulations. 
% \end{enumerate}

% The results are stated in Table~\ref{tab:stat-manipulation}.
% For each mechanism we report the fraction of students where our algorithm found a profitable manipulation, as well as the average utility gain for that manipulation.

\begin{table}[h]
\centering
\begin{tabular}{@{}cccccccccc@{}}
\toprule
\multirow{3}{*}{Instance} & \multirow{3}{*}{Uncertainty} & \multicolumn{8}{c}{Mechanism} \\ 
 &  & \multicolumn{2}{c}{No EF-TB} & \multicolumn{2}{c}{Classic EF-TB} & \multicolumn{2}{c}{Contested EF-TB} & \multicolumn{2}{c}{HBS} \\  \cmidrule(lr){3-4} \cmidrule(lr){5-6} \cmidrule(lr){7-8} \cmidrule(lr){9-10} 
 &  & \# & Gain & \# & Gain & \# & Gain & \# & Gain \\ \midrule
\multirow{2}{*}{Small} & Randomness & 1 & 0.04\% & 0 & - & 1 & 0.04\% & 66 & $8.0 \pm 2.1\%$ \\
 & Population & 0 & - & 0 & - & 0 & - & 63 & $9.0 \pm 2.1\%$ \\ 
\multirow{2}{*}{Ivy Small} & Randomness & 21 & $13.5 \pm 3.6\%$ & 15 & $8.5 \pm 2.2 \%$ & 1 & 0.02\% & 107 & $23.5 \pm 3.6 \%$ \\
 & Population & 20 & $7.4\pm 1.4\%$ & 11 & $4.7\pm 1.0\%$ & 0 & - & 117 & $20.7 \pm 3.0 \%$ \\
\multirow{2}{*}{Biz} & Randomness & 0 & - & 0 & - & 0 & - & 87 & $26.3 \pm 4.0\%$ \\
 & Population & 1 & 0.7\% & 0 & - & 0 & - &  64 & $21.6 \pm 2.8\%$ \\
\bottomrule
\end{tabular}
\caption{Number of statistically significant manipulations found by our algorithm and their average relative gain ($\pm$ standard error) in utility.}
\label{tab:stat-manipulation}
\end{table}

% \Ruiquan{Is this table better on top of the page (by changing ``[h]'' to ``[t]'')?}

\begin{remark}
We are not sure  why profitable manipulations for no EF-TB and classical EF-TB manipulations were found only for Ivy Small and not for the other schools. We observe  at equilibrium Ivy Small has two courses that are in very high demand --- their price exceeds the entire budgets of some students.
\end{remark}

\bibliographystyle{ACM-Reference-Format}
\bibliography{bib}

\appendix

\section{Additional related work}
\label{app:related}
% \subsection{Additional related work}

\cite{WV15} consider an adaptation of A-CEEI to team formation games (where agents need to be able to afford all their teammates). They empirically evaluate the gain from a manipulation using a different algorithm than the one we describe in Section~\ref{sec:manip}. They report that on an instance based on 17 students ranking each other, the adapted A-CEEI was significantly {\em more} manipulable than the adapted HBS.  In contrast, in our experiments for course allocation HBS is always more manipulable.

\cite{DABMS14} survey different mechanisms for the course allocation problem. A newer work by~\cite{MBU19} focuses course allocation via a variant of Probabilistic Serial mechanism;  incentives properties of Probabilistic Serial were also recently studied by~\cite{WWZ20}. 

\cite{BHM15} study the existence and complexity of exact CEEI with indivisible goods  which are either perfect substitutes or complements (unfortunately courses in our instances satisfy neither).
\cite{GHMS21-chores,BCM22,CGMM22-chores, CGMM22b-chores} study the problem of designing (approximate or exact) algorithms for computing a CEEI over divisible chores. Interesting variants of CEEI have also been studied by~\cite{Aziz15, BLM16, PK19, BNT21}.

The fair allocation of indivisible goods has received a lot of attention recently, see for example the recent surveys of~\cite{DBLP:reference/choice/BouveretCM16} and~\cite{AAB+22-survey}. In particular, there is a large body of theoretical work on existence, algorithms, and complexity of allocations satisfying EF1 (as does A-CEEI) or its strengthening EFX (e.g.~\cite{LMMS04,DBLP:conf/ijcai/BouveretCEIP17,BKV18,CGH19,PR20,DBLP:journals/tcs/AmanatidisBFHV21,CGMMM21,doi:10.1287/moor.2021.1196,SUKSOMPONG2023110956}. Most of those works focus on cases where the agents' valuation functions satisfy nice properties ranging from additivity to subadditivity; some of these assumptions are applicable in other settings (see e.g.~Spliddit~\cite{GP14}), but for course allocation real-world students utilities are quite complex --- they are not sub-additive, and in fact not even monotone.

There has also been a recent body of work on the computation of economic equilibria, which could be considered the most practical application of Alvin Roth's idea of ``The Economist as Engineer''~\cite{https://doi.org/10.1111/1468-0262.00335}.~\cite{RobustRentDivision}, similar to our work, provides theoretical grounding and experimental results based on real data for the fair allocation of indivisible goods in a specific context (rent sharing). Works such as~\cite{KC82,CheungCD20,LW20} (and many references therein) provide some theoretical grounding for what we found practically in developing our new search algorithm: that gradient descent-inspired tatonnement may be surprisingly effective for finding approximate equilibria despite worst-case complexity results.

% Here's the intro to "https://arxiv.org/pdf/2112.04166.pdf": Research in fair division is quickly moving from the realm of theory to practical applications, ranging from the division of various assets between individuals (Goldman and Procaccia 2014) to the distribution of food to charities (Aleksandrov et al. 2015) and medical equipment among communities (Pathak et al. 2021).

%\Abe{Should be strategic about including citations here. To find potential reviewers or literature touchstones I looked at the two most recent ECs for special material related to what we're doing. I found this workshop on fair division: https://sites.google.com/view/fairacac/program from EC21 as well as two workshops at EC22.}=

\section{The benchmark algorithm}\label{app:baseline}

%\Aviad{TBD: Discuss three stages}

Our benchmark for comparison is the (previous) commercial state of the art algorithm. We provide an overview of the algorithm here, and more details can be found in academic publications~\cite{OthmanSB10,BudishCKO17}%
\footnote{Our benchmark corresponds to the ``first phase'' of the algorithm described in~\cite{BudishCKO17}. The remaining two phases are only used to cope with the fact that in practice clearing error is one-sided: undersubscribed courses are undesirable, but oversubscribed courses are absolutely infeasible, due e.g.~to safety regulations for room capacities. This is a conservative comparison: we show that our algorithm is already faster than the first phase, and finds assignments with {\em zero} clearing error.}.

%of~\cite{BudishCKO17} consists of three stages. The first stage follows the tabu search heuristic of~\cite{OthmanSB10}, which is shown in Algorithm~\ref{alg:tabu-search}. After the first stage, the algorithm can find a solution with clearing error under the theoretical bound of~\cite{Budish2011}. Based on the solution, the second stage eliminates all over-subscription by increasing individual prices and the third stage improves the under-subscription by increasing the budgets of the students. Next, we shall introduce the tabu search heuristics of the first stage in more detail.

\begin{algorithm}[t]
    \caption{Tabu search}
    \label{alg:tabu-search}
    \begin{algorithmic}
    \State \textbf{Inputs:} courses' capacities $\vc{c}$, students' preferences $\vc{u}$, initial budgets $\vc{b}_0\in [1,1+\beta]^n$.
    \State \textbf{Parameters:} neighborhood function $\neighbors(\vc{p})$, binary relation $\sim_p$ for prices.
    \State \textbf{Algorithm:}
    \begin{enumerate}
      \item Let $\vc{p} \gets \text{uniform}(1,1+\beta)^m, \mathcal{H} \gets \emptyset$.
      \item If $\|\vc{z}(\vc{u}, \vc{c}, \vc{p}, \vc{b}_0)\|_2=0$, terminate with $\vc{p}^*=\vc{p}$.
      \item Otherwise, 
      \begin{itemize}
        \item include all equivalent prices of $\vc{p}$ into the history: $\mathcal{H}\gets \mathcal{H} + \{\vc{p'}: \vc{p'} \sim_p \vc{p}\}$,
        \item update $\vc{p}\gets \arg \min_{\vc{p'}\in \neighbors(\vc{p}) - \mathcal{H}} \|\vc{z}(\vc{u}, \vc{c}, \vc{p'}, \vc{b}_0)\|_2$, and then
        \item go back to step 2.
      \end{itemize}
    \end{enumerate}
    \end{algorithmic}
\end{algorithm}

The basic idea for both the benchmark and our algorithm is tatonnement: increase the price of over-demanded courses, and decrease the prices of under-demanded courses (see also Algorithm~\ref{alg:tatonnement} with $\epsilon = 0$). 
The benchmark (see~\ref{alg:tabu-search} for pseudocode) augments the vanilla tatonnement with {\em tabu search}. That is, in each iteration it explores multiple neighbors in the price space and then updates to the one with minimum clearing error. To avoid repetitive searches, it will not consider price vectors in the neighborhood that are equivalent to some explored price. More specifically, two prices are defined to be equivalent if the student demands are identical under the prices, i.e. for any price vectors $\vc{p}, \vc{p'}$, $\vc{p}\sim_p\vc{p'}$ if 
\begin{align*}
\forall \text{ student $i\in[n]$}, ~~\vc{a}_i(\vc{u}, \vc{p}, \vc{b_0}) = \vc{a}_i(\vc{u}, \vc{p'}, \vc{b_0})~.
\end{align*}
The neighborhood function $\neighbors(\vc{p})$ consists of the following two types of neighbors.
\begin{enumerate}
    \item \emph{Gradient neighbors.} In gradient neighbors, prices are adjusted in proportion to the excess demand, i.e., for a time-varying set $\Delta \subseteq [0,1]$, the gradient neighbors consist of the following prices: 
    \begin{align*}
        \mathcal{N}^{\text{gradient}}(\vc{p},\Delta) = \{\vc{p}+\delta \cdot \vc{z}(\vc{u},\vc{c}, \vc{p}, \vc{b}): \delta \in \Delta \}~.
    \end{align*}
    \item \emph{Individual price adjustment neighbors.} In individual price adjustment neighbors, the price vector is adjusted in only a small number of entries. For over-demanded courses, we consider increasing their prices so that their total demands decrease by one. For under-demanded courses, we consider setting them to zero or decreasing them by a fixed amount. To avoid each iteration exploring too many individual price adjustment neighbors, the algorithm randomly merges these individual price adjustments into 35 price updates.% \Ruiquan{We know this detail from Abe's email/code, should we mention it here?} Yeah, that's fine
\end{enumerate}

\section{Economic efficiency of our algorithm}
\label{app:econeff}
%\section{Comparing the Nash social welfare}

%\Aviad{Maybe we should also compare the (utilitarian) social welfare?}
% \Ruiquan{I found in one instance where the benchmark produces zero-error solutions, the benchmark can significantly beat us on both social welfares. The benchmark seems to have some implicit optimization for these. I would try to integrate some optimization for SW into our code and see what happens.}
% \Ruiquan{PS: I also find on Columbia Fall 2021, the benchmark cannot produce a non-zero NSW, but ours can!}

The focus of our work is improving the computational efficiency and exploring the incentive properties of the A-CEEI problem. However, we were also curious to see how our new algorithm compares to the old benchmark in terms of economic efficiency of the allocations. Specifically, we use weighted variants of social welfare and Nash social welfare. The weights correspond to  students' (unperturbed) budgets: while in the original A-CEEI all agents have equal incomes, in practice schools often give students different budgets (mostly based on seniority) that correspond to priorities.  

For this section we compare to the first-two-stages-only variant of the old benchmark: The first-stage-only benchmark may have oversubscribed courses, while  the full 3-stage benchmark is not known to satisfy fairness or incentives (e.g.~SP-L) properties, so those would not be fair benchmarks for comparisons. 

The results are summarized in Table~\ref{tbl:nsw-of-algos}. At a high level, our algorithm is somewhat better on both metrics on most (although not all) instances. 

\paragraph{Utilitarian Social Welfare.} We compare the utilitarian social welfare (USW) of our solutions and the solutions found by the benchmark. For any algorithm $A$, the average USW is defined to be 
\begin{align*}
\sum_{i\in [n]} \frac{b_i}{\sum_{i\in [n]}b_i} \cdot u_i(\vc{A}_i(\vc{u}, \vc{c}))~,
\end{align*}
where $b_i$ is the base budget of student $i$. 

\paragraph{Nash Social Welfare.} We also compare the nash social welfare (NSW) of our solutions and the solutions found by the whole baseline algorithm. For any algorithm $A$, the NSW is defined to be 
\begin{align*}
\prod_{i\in [n]} u_i(\vc{A}_i(\vc{u}, \vc{c}))^{b_i/\sum_{i\in [n]} b_i}~,
\end{align*}
where $b_i$ is the base budget of student $i$. %As NSW can be exponentially large, for simplicity, we calculate the log NSW (with base $e$). % and show the results are shown in Table~\ref{tbl:nsw-of-algos}. We can see that the NSW of the solutions found our algorithm improves that of the baseline by $??? \%$ to $??? \%$.

\begin{table}[t]
\centering
\begin{tabular}{ccccccc}
\toprule
\multirow{2}{*}{Instance} & \multicolumn{3}{c}{Avg. NSW}                                             & \multicolumn{3}{c}{Avg. USW}                                              \\ \cmidrule(lr){2-4}\cmidrule(lr){5-7}
                          & \multicolumn{1}{c}{Fast} & \multicolumn{1}{c}{Benchmark} & Imp. & \multicolumn{1}{c}{Fast} & \multicolumn{1}{c}{Benchmark} & Imp. \\ \midrule
      Ivy Huge Fall 21
 & \multicolumn{ 1 }{c}{470.44} & \multicolumn{ 1 }{c}{-} & $+\infty$
 & \multicolumn{ 1 }{c}{552.92} & \multicolumn{ 1 }{c}{541.5} & $2.11\%$
\\
Law Fall 19
 & \multicolumn{ 1 }{c}{846.66} & \multicolumn{ 1 }{c}{837.29} & $1.12\%$
 & \multicolumn{ 1 }{c}{1045.25} & \multicolumn{ 1 }{c}{1036.76} & $0.82\%$
\\
Law Fall 20
 & \multicolumn{ 1 }{c}{646.16} & \multicolumn{ 1 }{c}{572.66} & $12.83\%$
 & \multicolumn{ 1 }{c}{848.04} & \multicolumn{ 1 }{c}{764.39} & $10.94\%$
\\
Law Fall 21
 & \multicolumn{ 1 }{c}{168.93} & \multicolumn{ 1 }{c}{150.85} & $11.99\%$
 & \multicolumn{ 1 }{c}{194.15} & \multicolumn{ 1 }{c}{176.24} & $10.16\%$
\\
Ivy Small Fall 19
 & \multicolumn{ 1 }{c}{444.32} & \multicolumn{ 1 }{c}{445.86} & $-0.35\%$
 & \multicolumn{ 1 }{c}{567.31} & \multicolumn{ 1 }{c}{564.43} & $0.51\%$
\\
Ivy Small Fall 20
 & \multicolumn{ 1 }{c}{421.54} & \multicolumn{ 1 }{c}{415.86} & $1.37\%$
 & \multicolumn{ 1 }{c}{501.29} & \multicolumn{ 1 }{c}{496.28} & $1.01\%$
\\
Ivy Small Fall 21
 & \multicolumn{ 1 }{c}{437.49} & \multicolumn{ 1 }{c}{444.79} & $-1.64\%$
 & \multicolumn{ 1 }{c}{535.47} & \multicolumn{ 1 }{c}{540.59} & $-0.95\%$
\\
Ivy Small Fall 22
 & \multicolumn{ 1 }{c}{462.49} & \multicolumn{ 1 }{c}{462.77} & $-0.06\%$
 & \multicolumn{ 1 }{c}{521.4} & \multicolumn{ 1 }{c}{522.94} & $-0.29\%$
\\
Biz Fall 18
 & \multicolumn{ 1 }{c}{203.69} & \multicolumn{ 1 }{c}{202.75} & $0.46\%$
 & \multicolumn{ 1 }{c}{232.53} & \multicolumn{ 1 }{c}{231.77} & $0.33\%$
\\
Biz Fall 19
 & \multicolumn{ 1 }{c}{191.53} & \multicolumn{ 1 }{c}{189.0} & $1.34\%$
 & \multicolumn{ 1 }{c}{222.21} & \multicolumn{ 1 }{c}{219.88} & $1.06\%$
\\
Biz Fall 20
 & \multicolumn{ 1 }{c}{209.41} & \multicolumn{ 1 }{c}{191.97} & $9.08\%$
 & \multicolumn{ 1 }{c}{227.47} & \multicolumn{ 1 }{c}{212.09} & $7.25\%$
\\
Biz Fall 21
 & \multicolumn{ 1 }{c}{234.95} & \multicolumn{ 1 }{c}{234.62} & $0.14\%$
 & \multicolumn{ 1 }{c}{240.77} & \multicolumn{ 1 }{c}{240.51} & $0.11\%$
\\ \bottomrule
\end{tabular}
\caption{Comparing the (nash) social welfare.}
\label{tbl:nsw-of-algos}
\end{table}

% s3-rotman-fall-2018 232334.47 203455.3
% s3-rotman-fall-2019 219778.82 188760.36
% s3-rotman-fall-2020 222929.9 201402.02
% s3-rotman-fall-2021 240749.46 234924.51
% s3-rotman-fall-2018 ours! 232527.13 203690.33
% s3-rotman-fall-2019 ours! 222205.77 191525.91
% s3-rotman-fall-2020 ours! 227468.77 209406.06
% s3-rotman-fall-2021 ours! 240774.01 234946.89

\begin{comment}
\begin{table}[t]
\begin{tabular}{|c|c|c|c|c|}
    \hline
    Instance & NSW of ours & Errors of~\cite{BudishCKO17} & NSW of~\cite{BudishCKO17} & Improvements\\
    \hline
    Toronto Law in Fall 2019 & 846414.58 & 11 to 14 & 818017.24 & 3.47\% \\
    \hline 
    Toronto Law in Fall 2020 & 643539.84 & 26 to 89 & 582602.70 & 10.46\% \\
    \hline 
    Toronto Law in Fall 2021 & 168948.48 & 53 to 132 & 160694.91 & 5.14\% \\
    \hline
    Rotman in Spring 2019 & 201733.40 & 0 to 5 & 201645.67 &  0.04\%\\
    \hline 
    Rotman in Summer 2022 & 212964.54 & 0 & 202567.00 & 5.13\%\\
    \hline
\end{tabular}
\end{table}
\end{comment}

\section{Hardness of A-CEEI with real-world utility functions}

\label{app:hard}
In this appendix, we show finding an A-CEEI continues to be computationally intractable even in the special case of utilities as described in Eq.~\eqref{eq:original-utilities}. 

\begin{theorem}
There exists a constant $\beta>0$ such that finding $(\alpha,\beta)$-CEEI is \PPAD-complete for $\alpha=\Omega(n/m)$, even when the utilities are restricted to the form given in Eq.~\eqref{eq:original-utilities}. 
\end{theorem}

Our proof reduces the generalized circuit problem to A-CEEI with utilities as in Eq.~\eqref{eq:original-utilities}.
The reduction consists of two steps.
The first step is to reduce the generalized circuit problem~\cite{Rub18} to one of its simplifications.
In the simplified problem, we need to find an approximate solution of a (possibly cyclic) circuit only involving the SUM and NOT gates (Defintion~\ref{def:simp-gcircuit}), which are two out of the eight gates for the generalized circuit problem (Definition~\ref{def:gcircuit}).
To show the \PPAD-hardness of the simplified problem (Theorem~\ref{thm:simp-gcircuit-ppad-hard}), we construct each of the remaining six gates by combinations of SUM and NOT gates.
The second step is to reduce the simplified generalized circuit to A-CEEI.
Our reduction 
follows the framework of~\cite{OthmanPR16}.
We construct gadgets for the SUM and NOT gates (Lemma~\ref{lem:not-sum-gadget}).
The gadgets contain \emph{input courses}, an \emph{output course}, \emph{interior courses} and some students. 
An \emph{output course} of a gadget can be \emph{input courses} of other gadgets.
In each gadget, the prices of the \emph{input courses} and the \emph{output course} respectively correspond to the values of the input variables and the output variable in the circuit. 
Our construction guarantees that the courses prices approximately satisfy the constraints of the variables of their corresponding variables. 
However, these gadgets cannot be directly integrated into a cyclic circuit. 
For each gadget, we require an upper bound for the number of additional students that want the \emph{output course}.
In the construction, the upper bound for the output course is strictly smaller than that of the input courses in the gadget. 
If the circuit is cyclic, there can be contradictions on magnitude of the upper bounds. 
To resolve this issue,~\cite{OthmanPR16} introduces \emph{course-size amplification gadget}, a COPY gadget in which the price of the output course approximately equals the price of the input course and the upper bound of the output course is twice as large as that of the input course. 
And finally, we finish the proof by constructing the \emph{course-size amplification gadget} under utilies as in Eq.~\eqref{eq:original-utilities}  (Lemma~\ref{lem:course-size-amplification-gadget}). 

%\Aviad{We should formally state the main theorem for this section and write a short proof using all the lemmata. 
%In particular we should explain the amplification gadget is necessary.}

% \subsubsection{Simplified Generalized Circuit}
\subsection{Simplified Generalized Circuit}

In this subsection, we present our simplified generalized circuit and the \PPAD-hardness solving it. We start from the definition of the generalized circuit problem as in~\cite{OthmanPR16}. %\Aviad{This definition of the gcircuit is already slightly non-standard (although easily equivalent to the original)} \Ruiquan{I add a remark after the definition.} 

\begin{definition}[Generalized circuits]
    % \label{def:gcircuit}
    A generalized circuit $\mathcal{S}=(V,\mathcal{T})$ is a collection of nodes $V$ and gates $\mathcal{T}$. Each gate $T \in \mathcal{T}$ is a 4-tuple $(G,v_1,v_2,v)$, consisting of a type of gate $ G \in \{ G_{/2}, G_{\frac{1}{2}}, G_{+}, G_{-}, G_{<}, G_{\land}, G_{\lor}, G_{\lnot} \} $, two input nodes $v_1,v_2\in V\cup\{nil\}$ and an output node $v$.
    Each node $v\in V$ may be the output node for at most one gate, i.e., for every two gates $T=(G,v_1,v_2,v)$ and $T'=(G',v_1',v_2',v')$ in $\mathcal{T}$, $v\ne v'$. 
\end{definition}

\begin{definition}[$\epsilon$-\textsc{Gcircuit} problem]
\label{def:gcircuit}
    Given a generalized circuit $\mathcal{S}=(V,\mathcal{T})$ and an assignment $\vc{x}: V \to \mathbb{R}$, we say $\vc{x}$ \emph{$\epsilon$-approximately satisfies} $\mathcal{S}$ if $\vc{x}(v)\in[0,1+\epsilon]$ for any $v \in V$, and for each gate $T=(G,v_1,v_2,v) \in \mathcal{T}$, we have that $|\vc{x}(v)-f_G(\vc{x}(v_1),\vc{x}(v_2))|\le \epsilon$, where $f$ is defined as follows:
    \begin{enumerate}
        \item HALF: $f_{G_{/2}}(x)=x/2$
        \item VALUE: $f_{G_{\frac{1}{2}}}\equiv\frac{1}{2}$
        \item SUM: $f_{G_{+}}(x,y)=\min(x+y,1)$
        \item DIFF: $f_{G_{-}}(x,y)=\max(x-y,0)$
        \item LESS: $f_{G_{<}}(x,y)=\begin{cases}
        1 & x>y+\epsilon\\
        0 & y>x+\epsilon
        \end{cases}$
        \item AND: $f_{G_{\land}}(x,y)=\begin{cases}
        1 & (x>0.9+\epsilon)\land(y>0.9+\epsilon)\\
        0 & (x<0.1-\epsilon)\lor(y<0.1-\epsilon)
        \end{cases}$
        \item OR: $f_{G_{\lor}}(x,y)=\begin{cases}
        1 & (x>0.9+\epsilon)\lor(y>0.9+\epsilon)\\
        0 & (x<0.1-\epsilon)\land(y<0.1-\epsilon)
        \end{cases}$
        \item NOT: $f_{G_{\lnot}}(x)=1-x$.
    \end{enumerate}
\end{definition}

\begin{remark}
Definition~\ref{def:gcircuit} is slightly different from those in~\cite{OthmanPR16} and~\cite{Rub18} on AND and OR gates. It can be observed that this problem (with fan-out 2) is harder than that in~\cite{Rub18} for sufficiently small constant $\epsilon>0$ and is thus \PPAD-complete for some constant $\epsilon>0$.
\end{remark}

Our simplification is that we restrict the set of types of gates to only consist of SUM and NOT. Formally, the definition of the circuit and the problem is as follows.

\begin{definition}[Simplified Generalized Circuit]
    % \label{def:simp-gcircuit}
    A simplified generalized circuit $\mathcal{S}=(V,\mathcal{T})$ is a collection of nodes $V$ and gates $\mathcal{T}$. Each gate $T \in \mathcal{T}$ is a 4-tuple $(G,v_1,v_2,v)$, consisting of a type of gate $ G \in \{ G_{+}, G_{\lnot} \} $, two input nodes $v_1,v_2\in V\cup\{nil\}$ and an output node $v$.
    Each node $v\in V$ may be the output node for at most one gate, i.e., for every two gates $T=(G,v_1,v_2,v)$ and $T'=(G',v_1',v_2',v')$ in $\mathcal{T}$, $v\ne v'$. 
\end{definition}

\begin{definition}[Simplified $\epsilon$-\textsc{Gcircuit} problem]
    \label{def:simp-gcircuit}
    Given a simplified generalized circuit $\mathcal{S}=(V,\mathcal{T})$ and an assignment $\vc{x}: V \to \mathbb{R}$, we say $\vc{x}$ \emph{$\epsilon$-approximately satisfies} $\mathcal{S}$ if $\vc{x}(v)\in[0,1+\epsilon]$ for any $v \in V$, and for each gate $T=(G,v_1,v_2,v) \in \mathcal{T}$, we have that $|\vc{x}(v)-f_G(\vc{x}(v_1),\vc{x}(v_2))|\le \epsilon$, where $f$ is defined as follows:
    \begin{enumerate}
        \item SUM: $f_{G_{+}}(x,y)=\min(x+y,1)$
        \item NOT: $f_{G_{\lnot}}(x)=1-x$.
    \end{enumerate}
\end{definition}

In the remaining part of this subsection, we will prove the following hardness result:
\begin{theorem}
    \label{thm:simp-gcircuit-ppad-hard}
There exists a constant $\epsilon>0$ such that simplified $\epsilon$-\textsc{Gcircuit} problem with fan-out 2 is \PPAD-complete.
\end{theorem}
More specifically, we will prove the hardness result by showing that a $\epsilon$-\textsc{Gcircuit} problem can be reduced to a simplified $\epsilon_0=\epsilon/11.5$-\textsc{Gcircuit} problem. Next, we present how to approximately construct the other 6 gates in generalized circuits by NOT and SUM gates that $\epsilon_0$-approximately satisfy SUM and NOT. For convenience, we define a NOT-SUM gadget that applies SUM on the two input variables and then applies NOT on the resulting variable. It is clear that we can approximate such gadgets with errors less than $2\epsilon_0$.  

\paragraph{DIFF gates.} We implement 
$$v_1 - v_2 = \lnot \Big((\lnot v_1) +v_2 \Big).$$% f_{G_{\lnot}}(f_{G_{+}}(f_{G_{\lnot}}(v_1),v_2)).$$ 
In more detail, notice that with the NOT gate, we can construct a variable $\overline{v_1}$ with value $\vc{x}(\overline{v_1})\in [1-\vc{x}(v_1)-\epsilon_0, 1-\vc{x}(v_1)+\epsilon_0]$. By applying NOT-SUM on $\overline{v_1}$ and $v_2$, we can construct a variable $v$ such that 
\begin{align*}
    \vc{x}(v) &\in [1-\vc{x}(\overline{v_1})-\vc{x}(v_2)-2\epsilon_0, 1-\vc{x}(\overline{v_1})-\vc{x}(v_2)+2\epsilon_0]\\
    &\subseteq [\vc{x}(v_1)-\vc{x}(v_2)-3\epsilon_0, \vc{x}(v_1)-\vc{x}(v_2)+3\epsilon_0]~,
\end{align*}
which approximates the output with an error less than $3\epsilon_0$. 

\paragraph{HALF and VALUE gates.} 
Taking advantage of the fact that we have a {\em generalized} circuit, we observe that $v = v_1/2$ is the unique solution 
$$v = v_1 - v .$$
% \Aviad{Is the above statement that I wrote true? If so maybe we don't need 3 variables?}
% \Ruiquan{Yes. Maybe we can also directly use DIFF gates to simplify the implementation?}
In more detail, notice that with the previously constructed DIFF gates, we can construct a variable $v$ such that 
\begin{align*}
    \vc{x}(v) \in \vc{x}(v_1)-\vc{x}(v)\pm 3\epsilon_0\Rightarrow \vc{x}(v) \in \frac{1}{2}\vc{x}(v_1) \pm 1.5\epsilon_0~,
\end{align*}
which approximates the output of the HALF gate with an error less than $1.5\epsilon_0$. 

% we use the NOT gate to construct $\overline{v_1}$ with value $\vc{x}(\overline{v_1})\in [1-\vc{x}(v_1)-\epsilon_0, 1-\vc{x}(v_1)+\epsilon_0]$. With NOT-SUM gadgets, there can be three variables $v, v',v''$ such that
% \begin{align}
%     \label{eqn:half-gadget}
%     \vc{x}(v) &\in [1-\vc{x}(\overline{v_1})-\vc{x}(v')-2\epsilon_0, 1-\vc{x}(\overline{v_1})-\vc{x}(v')+2\epsilon_0]~,\notag\\
%     \vc{x}(v') &\in [1-\vc{x}(\overline{v_1})-\vc{x}(v'')-2\epsilon_0, 1-\vc{x}(\overline{v_1})-\vc{x}(v'')+2\epsilon_0]~,\\
%     \vc{x}(v'') &\in [1-\vc{x}(\overline{v_1})-\vc{x}(v)-2\epsilon_0, 1-\vc{x}(\overline{v_1})-\vc{x}(v)+2\epsilon_0]~.\notag
% \end{align}
% Therefore, given input variable $v_1$, $v$ can be the output variable of a HALF gate with an error less than $4\epsilon_0$, because
% \begin{align*}
%     \vc{x}(v), \vc{x}(v'), \vc{x}(v'') &\in \bigg[\frac{1-\vc{x}(\overline{v_1})}{2}-3\epsilon_0, \frac{1-\vc{x}(\overline{v_1})}{2}+3\epsilon_0\bigg]\\
%     &\subseteq \bigg[\frac{\vc{x}(v_1)}{2}-4\epsilon_0, \frac{\vc{x}(v_1)}{2}+4\epsilon_0\bigg]~.
% \end{align*}

Similarly, we observe that $v=\frac{1}{2}$ is the unique solution of 
$$v=1-v~.$$
In more detail, notice that with NOT gates, we can construct a variable $v$ such that $\vc{x}(v)\in 1-\vc{x}(v)\pm \epsilon_0$ and approximates the output of the VALUE gate with error less than $0.5\epsilon_0$. 
% if we remove all the $\vc{x}(\overline{v_1})$ terms in~\eqref{eqn:half-gadget}, $v$ can be the output variable of a VALUE gate with error less than $3\epsilon_0$.

\paragraph{ROUND gadgets.} We shall introduce the ROUND gadget, which will be used for the construction of the remaining AND, OR, LESS gates. The ROUND gadget has one input variable and one output variable, and approximates the following function:
\begin{align*}
    f_{G_{[]}}(x) = \begin{cases}
        0 & x\leq 0.25~,\\
        1 & x\geq 0.75~.
    \end{cases}
\end{align*}

With a VALUE and a HALF gate, there can be a variable $v_{1/4}$ with value $0.25\pm 2\epsilon_0$. By applying a DIFF gate on $v_1,v_{1/4}$, we can get a variable $v_1^-$ such that 
\begin{align*}
    \vc{x}(v_1^-) \in \begin{cases}
        0 \pm 3.5\epsilon_0 & \vc{x}(v_1) \leq 0.25 \\
        0.5\pm 3.5\epsilon_0 & \vc{x}(v_1) \geq 0.75
    \end{cases}
\end{align*}
Therefore, using a SUM gate with both inputs $v_1^-$, we can approximate the ROUND gadget with error less than $8\epsilon_0$. 

\paragraph{AND and OR gates.} We shall first show the construction of NOR gadgets, which approximates the following function 
\begin{align*}
    f_{G_{\downarrow}}(x,y) = \begin{cases}
        0 & (x>0.9)\lor(y>0.9)\\
        1 & (x<0.1)\land(y<0.1)
    \end{cases}
\end{align*}
We can construct them with errors less than $8\epsilon_0$ by applying a NOT-SUM gadget on $v_1,v_2$ and then applying a ROUND gadget on the resulting variable. To construct OR gates, we can first apply a NOR gadget and then apply a NOT gate. Then, the resulting variable $v$ satisfies
\begin{align*}
    \vc{x}(v) \in \begin{cases}
        1 \pm 9\epsilon_0 & (\vc{x}(v_1)>0.9)\lor(\vc{x}(v_2)>0.9)\\
        0 \pm 9\epsilon_0 & (\vc{x}(v_1)<0.1)\land(\vc{x}(v_2)<0.1)
    \end{cases}
\end{align*}
which approximates OR gates with errors less than $9\epsilon_0$. To construct AND gates, we can first apply two NOT gates respectively on $v_1,v_2$ and then apply the NOR gadget. Then, the resulting variable $v$ satisfies
\begin{align*}
    \vc{x}(v) \in \begin{cases}
        1\pm 8\epsilon_0 & (\vc{x}(v_1)>0.9+\epsilon_0)\land(\vc{x}(v_2)>0.9+\epsilon_0)\\
        0 \pm 8\epsilon_0 & (\vc{x}(v_1)<0.1-\epsilon_0)\lor(\vc{x}(v_2)<0.1-\epsilon_0)
    \end{cases}
\end{align*}
which approximates AND gates with errors less than $8\epsilon_0$.

\paragraph{LESS gates.} With a DIFF gate on $v_1$ and $v_2$, we can get a variable $\tilde{v}$ such that 
\begin{align*}
    \vc{x}(\tilde{v}) \in \begin{cases}
        (10\epsilon_0, +\infty) & \vc{x}(v_1)>\vc{x}(v_2)+11.5\epsilon_0\\
        (-\infty, 1.5\epsilon_0) & \vc{x}(v_1)<\vc{x}(v_2)-11.5\epsilon_0
    \end{cases}
\end{align*}

Next, we consider a DOUBLE gadget, which is implemented by a SUM gate with both input variables. Because the error of SUM is $\epsilon_0$, the error of DOUBLE is also $\epsilon_0$. Then, we can derive the following lemma.
\begin{lemma}
After applying DOUBLE gadgets on a variable $v$ for $k$ times, the resulting variable $v'$ satisfies:
\begin{align*}
    \vc{x}(v') \in [2^k \vc{x}(v) - (2^k-1)\epsilon_0, 2^k \vc{x}(v) + (2^k-1)\epsilon_0]~.
\end{align*}
\end{lemma}
\begin{proof}
    We shall prove it by induction. The base case is trivial when $k=0$. Suppose the statement holds for $k$. After applying the $(k+1)$th DOUBLE gadget, the resulting value lies in the interval:
    \begin{align*}
        % [DOUBLE(2^k x - (2^k-1)\epsilon_0), DOUBLE(2^k x + (2^k-1)\epsilon_0)] &\subseteq
         [2(2^k \vc{x}(v) - (2^k-1)\epsilon_0)-\epsilon_0&, 2(2^k \vc{x}(v) + (2^k-1)\epsilon_0) + \epsilon_0] \\
         &= 
         [2^{k+1} \vc{x}(v) - (2^{k+1}-1)\epsilon_0, 2(2^{k+1} \vc{x}(v) + (2^{k+1}-1)\epsilon_0)]~.
    \end{align*}
\end{proof}

Therefore, by applying DOUBLE gadget $-\log(10\epsilon_0)$ times on $\tilde{v}$, we can get a variable $\tilde{v}'$ such that:
\begin{align*}
    \vc{x}(\tilde{v}'') \in \begin{cases}
        (0.9+\epsilon_0, +\infty) & \vc{x}(v_1)>\vc{x}(v_2)+11.5\epsilon_0\\
        (-\infty, 0.25-\epsilon_0) & \vc{x}(v_1)<\vc{x}(v_2)-11.5\epsilon_0
    \end{cases}
\end{align*}

Further, with a ROUND gadget on it, we can construct LESS gates with errors less than $8\epsilon_0$. 

% \subsubsection{Reduce Simpified \textsc{Gcircuit} to A-CEEI}
\subsection{Reduce Simpified \textsc{Gcircuit} to A-CEEI}

In the reduction, we shall use a two-credit auxiliary course $c_0$ with infinite capacity. According to the definition of A-CEEI, in any $(\alpha,\beta)$-CEEI, its price $p_0$ must be zero. For each course $i$, let $n_i$ be the upper bound for the number of students wanting course $i$ in gadgets where $i$ is not the \emph{output course}.

We first present the construction of the following NOT-SUM gadget, which can be easily used to construct NOT and SUM gates.

\begin{lemma}[NOT-SUM gadget]
    \label{lem:not-sum-gadget}
    Let $n_j \geq n_i > 3\alpha$ and suppose that the economy contains the following courses:
    \begin{itemize}
        \item input course $c_i, c_j$;
        \item output course $c_k$ with capacity $q_k=2n_i/3$ (at most $n_k=n_i/3$ other students want $c_k$);
        \item interior course $c_0$;
    \end{itemize}
    and $n_i$ following students:
    \begin{itemize}
        \item 1 conflict: course $c_k$ and $c_0$ can not be simultaneously allocated to a student, i.e. $valid(\vc{x})=1$ if and only if $x_0+x_k\leq 1$;
        \item 1 requirement: the total allocated credits among courses $c_i,c_j,c_0$ is at least 2, i.e. $req(\vc{x})=1$ if and only if $x_i+x_j+2x_0\geq 2$;
        \item the utility function is:
        \begin{align*}
            u(\vc{x}) = \begin{cases}
                8x_i + 4x_j + 2x_k + x_0 + 16 & valid(\vc{x})=1, req(\vc{x}) = 1~,\\
                \hfill 8x_i + 4x_j + 2x_k + x_0 \hfill & valid(\vc{x})=1, req(\vc{x}) = 0~,\\
                \hfill -\infty\hfill & valid(\vc{x}) = 0~.
            \end{cases}
        \end{align*}
    \end{itemize}
    Then in any $(\alpha,\beta)$-CEEI,
    $$p_k^* \in [1-p_i^*-p_j^*,1-p_i^*-p_j^*+\beta]~.$$
\end{lemma}

\begin{proof}
It is easy to observe that the students' preferences are as follows:
\begin{align*}
    \{c_i,c_j,c_k\} \succcurlyeq \{c_i,c_j,c_0\} \succcurlyeq &\{c_i,c_j\} \succcurlyeq \{c_i,c_0\} \succcurlyeq \{c_j,c_0\} \succcurlyeq \{c_0\}\\
    \succcurlyeq &\{c_i,c_k\} \succcurlyeq \{c_i\} \succcurlyeq \{c_j,c_k\} \succcurlyeq \{c_j\} \succcurlyeq \{c_k\} \succcurlyeq \emptyset
\end{align*}
Because course $c_0$ is always affordable, the students must get a bundle that is worse than $\{c_0\}$. Next, we discuss all possibilities of prices $p^*_i,p^*_j,p^*_k$ of $c_i,c_j,c_k$ in any $(\alpha,\beta)$-CEEI.
\begin{itemize}
    \item If $p^*_i+p^*_j+p^*_k<1$, the students will be all allocated the bundle $\{c_i,c_j,c_k\}$. Therefore, there are at least $n_i$ students allocated $c_k$, and thus the excess demand of $c_k$ is at least $n_i/3>\alpha$, which is impossible in any $(\alpha,\beta)$-CEEI. 
    \item If $p^*_i+p^*_j+p^*_k>1+\beta$, because the bundle $\{c_0\}$ is always affordable and better than any bundle involving $c_k$, no students will be allocated $c_k$. Therefore, there are at most $n_i/3$ students allocated $c_k$, and thus the excess demand of $c_k$ is at most $-n_i/3<-\alpha$, which is only possible in any $(\alpha,\beta)$-CEEI when $p^*_k=0$. 
\end{itemize}
Hence, in any $(\alpha,\beta)$-CEEI,
    $$p_k^* \in [1-p_i^*-p_j^*,1-p_i^*-p_j^*+\beta]~.$$
\end{proof}

Note that this gadget cannot work when the output course is identical to any of the input courses. To fix this issue, for any (possibly identical) input variables $i,j$ and output variable $k$, we can use two additional auxiliary courses $c_x,c_y$ and three above gadgets as follows:
\begin{itemize}
    \item input $c_i,c_j$ and output $c_x$, after which $p^*_x\in [1-p^*_i-p^*_j, 1-p^*_i-p^*_j+\beta], n_x=n_i/3$;
    \item input $c_x,c_0$ and output $c_y$, after which $p^*_y\in [1-p^*_x,1-p^*_x+\beta]\subseteq [p^*_i+p^*_j-\beta, p^*_i+p^*_j+\beta], n_y=n_x/3=n_i/9$;
    \item input $c_y,c_0$ and output $c_k$, after which $p^*_k\in [1-p^*_y, 1-p^*_y+\beta] \subseteq [1-p^*_i-p^*_j-\beta, 1-p^*_i-p^*_j+2\beta], n_k=n_y/3=n_i/27$.
\end{itemize}

However, because $n_k<n_i$ in the above construction, the gadgets cannot be directly integrated into cyclic circuits. To integrate the gadgets, we can use the following \emph{course-size amplification gadget}, which can approximately preserve the price of the input course and allow twice as many additional students to want the output course. 

\begin{lemma}[Course-size amplification gadget]
    \label{lem:course-size-amplification-gadget}
    Let $n_i > 2\alpha$ and suppose that the economy contains the following courses:
    \begin{itemize}
        \item input course $c_i$;
        \item output course $c_{j}$ with capacity $q_{j}=2.5n_i$ (at most $n_{j}=2n_i$ other students want $c_{j}$);
        \item interior course $c_0$, and course $c_1$ with capacity $q_1=3.5n_i$;
    \end{itemize}
    and the following sets of students:
    \begin{itemize}
        \item $n_i$ students with:
        \begin{itemize}
            \item 1 conflict: course $c_0$ and $c_1$ can not be simultaneously allocated to a student, i.e. $valid(\vc{x})=1$ if and only if $x_0+x_1\leq 1$;
            \item 1 requirement: the total number of allocated courses among $c_i,c_0$ is at least 1, i.e. $req(\vc{x})=1$ if and only if $x_i+x_0\geq 1$;
            \item the utility function is:
            \begin{align*}
                u(\vc{x}) = \begin{cases}
                    4x_i + 2x_1 + x_0 + 8 & valid(\vc{x})=1, req(\vc{x}) = 1~,\\
                    \hfill 4x_i + 2x_1 + x_0 \hfill & valid(\vc{x})=1, req(\vc{x}) = 0~,\\
                    \hfill -\infty\hfill & valid(\vc{x}) = 0~.
                \end{cases}
            \end{align*}
        \end{itemize}
        \item $n_1=3n_i$ students with: 
        \begin{itemize}
            \item 1 conflict: course $c_0$ and $c_j$ can not be simultaneously allocated to a student, i.e. $valid(\vc{x})=1$ if and only if $x_0+x_j\leq 1$;
            \item 1 requirement: the total number of allocated courses among $c_0,c_1$ is at least 1, i.e. $req(\vc{x})=1$ if and only if $x_0+x_1\geq 1$;
            \item the utility function is:
            \begin{align*}
                u(\vc{x}) = \begin{cases}
                    4x_1 + 2x_j + x_0 + 8 & valid(\vc{x})=1, req(\vc{x}) = 1~,\\
                    \hfill 4x_1 + 2x_j + x_0 \hfill & valid(\vc{x})=1, req(\vc{x}) = 0~,\\
                    \hfill -\infty\hfill & valid(\vc{x}) = 0~.
                \end{cases}
            \end{align*}
        \end{itemize}
    \end{itemize}
    Then in any $(\alpha,\beta)$-CEEI,
    $$p_{j}^* \in [p_i^*-\beta,p_i^*+\beta]~.$$
\end{lemma}
    
\begin{proof}
    It is easy to observe that the preferences of the first $n_i$ students and the last $n_1$ students are respectively as follows:
    \begin{align*}
        &\{c_i,c_1\} \succcurlyeq \{c_i,c_0\}\succcurlyeq\{c_i\}\succcurlyeq\{c_0\} \succcurlyeq \{c_1\} \succcurlyeq \emptyset~,\\
        &\{c_1,c_j\} \succcurlyeq \{c_1,c_0\}\succcurlyeq\{c_1\}\succcurlyeq\{c_0\} \succcurlyeq \{c_j\} \succcurlyeq \emptyset~.
    \end{align*}
    Because course $c_0$ is always affordable, the students must get a bundle that is no worse than $\{c_0\}$. Next, we start by discussing the prices $p^*_i,p^*_1$ of $c_i,c_1$ in any $(\alpha,\beta)$-CEEI.
    \begin{itemize}
        \item If $p^*_i+p^*_1<1$, the first $n_i$ students will be all allocated $\{c_i,c_1\}$. Further, because $p^*_1<1$, the other $3n_i$ students will all be allocated course $c_1$. Therefore, there are $4n_i$ students allocated $c_1$, and thus the excess demand of $c_1$ is $n_i/2>\alpha$, which is impossible in any $(\alpha,\beta)$-CEEI. 
        \item If $p^*_i+p_1>1+\beta$, none of the first $n_i$ students will be allocated course $c_1$. Therefore, there are at most $3n_i$ students allocated $c_1$, and thus the excess demand of $c_1$ is at most $-n_i/2<-\alpha$, which is impossible in any $(\alpha,\beta)$-CEEI.
    \end{itemize}
    Hence, $p_1^*\in [1-p_i^*,1-p_i^*+\beta]$. Finally, we shall show $p_{j}^* \in [1-p_1^*,1-p_1^*+\beta]$ to finish the proof.
    \begin{itemize}
        \item If $p_{j}^* + p_1^* > 1 + \beta$, none of the last $n_1$ students will be allocated course $c_j$. Therefore, there are at most $n_{j}$ students allocated $c_j$, and thus the excess demand of $c_j$ is at most $-n_i/2<-\alpha$. Because $p_1^*\leq 1+\beta$ in any $(\alpha,\beta)$-CEEI, $p_j^*>0$ and this is impossible in any $(\alpha,\beta)$-CEEI.
        \item If $p_j^*+p_1^*<1$, all the last $n_1$ students will be allocated the bundle $\{c_{1},c_j\}$. Therefore, there are at least $3n_i$ students allocated $c_j$, and thus the excess demand of $c_j$ is at least $n_i/2>\alpha$, which is impossible in any $(\alpha,\beta)$-CEEI. \qedhere
    \end{itemize}
\end{proof}

\section{Supplementary figures}
\label{app:figs}
\begin{figure}[H]
\centering
% \includegraphics[width=.49\textwidth]{img/CMP_compressed_big_time.pdf}
% \,
% \includegraphics[width=.49\textwidth]{img/CMP_compressed_big_iter.pdf}
% \\
% \includegraphics[width=.45\textwidth]{img_errors/s3-toronto-law-fall-19-robust-errors-wrt-time.pdf}
% \,
\includegraphics[width=.48\textwidth]{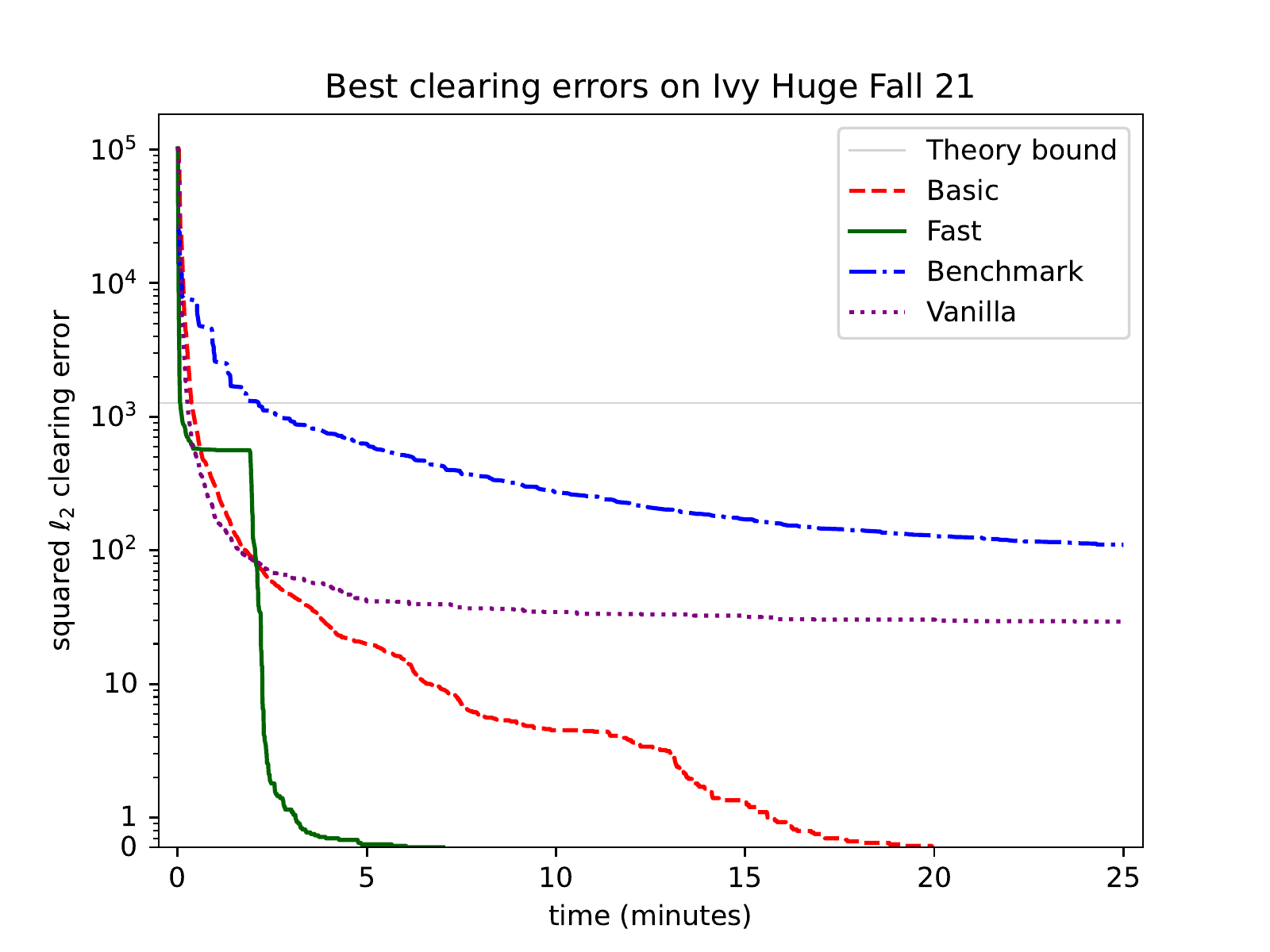}
\includegraphics[width=.48\textwidth]{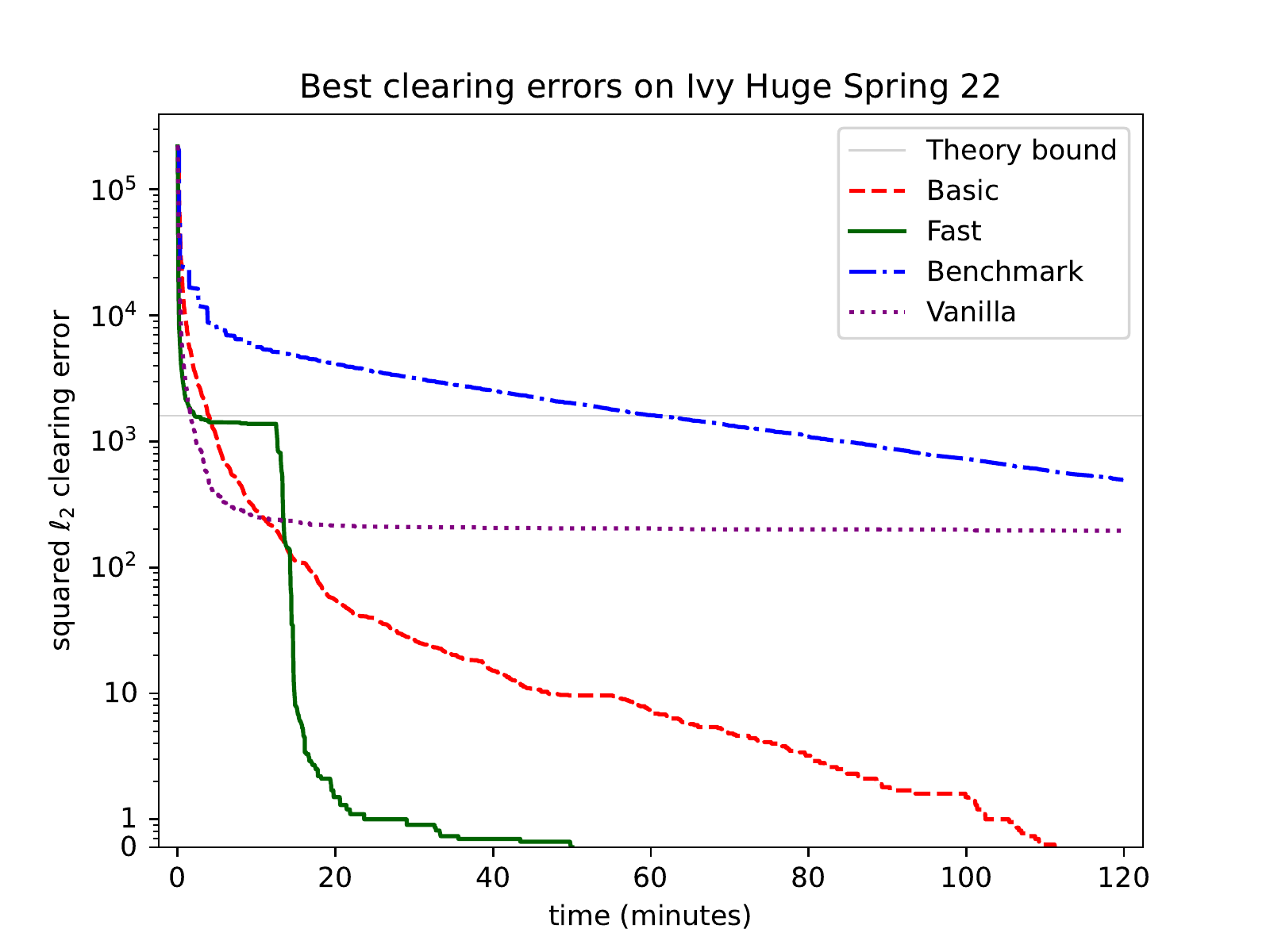}
\\
\includegraphics[width=.48\textwidth]{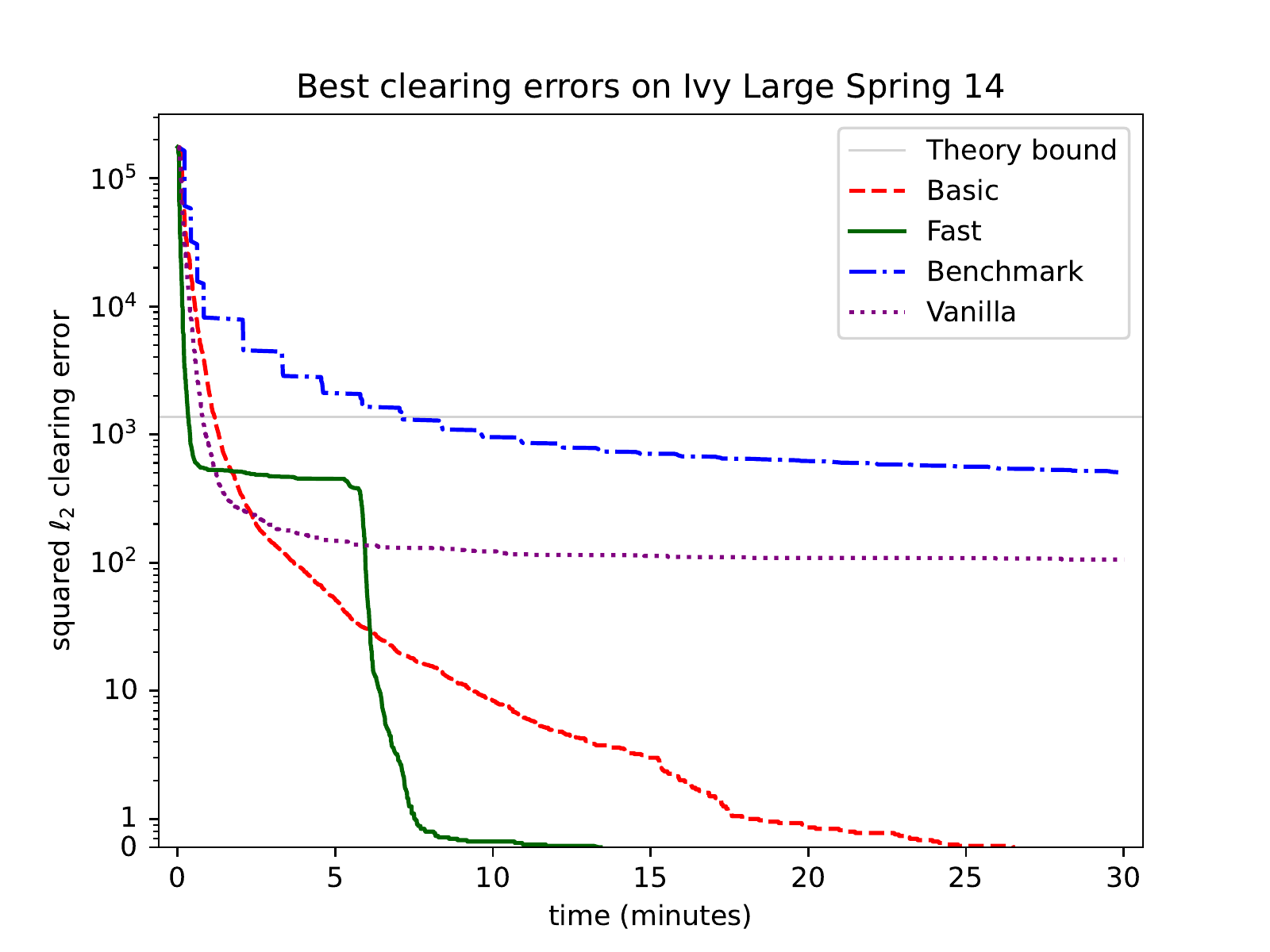}
\caption{For each of the algorithms tested and time $t$, we plot the best clearing error obtained by the algorithm up to time $t$.}
\label{fig:progress-of-algos-error-app}
\end{figure}

\begin{figure}[H]
    % \includegraphics[width=.32\textwidth]{img_warmstart/s3-toronto-law-fall-19-normal-distances-wrt-time.pdf}
    % \,
    % \includegraphics[width=.32\textwidth]{img_warmstart/large-normal-distances-wrt-time.pdf}
    % \,
    % \includegraphics[width=.32\textwidth]{img_warmstart/giant-normal-distances-wrt-time.pdf}
    % \\
    \centering
    % \includegraphics[width=.45\textwidth]{img_warmstart/s3-toronto-law-fall-19-robust-distances-wrt-time.pdf}
    % \,
    \includegraphics[width=.48\textwidth]{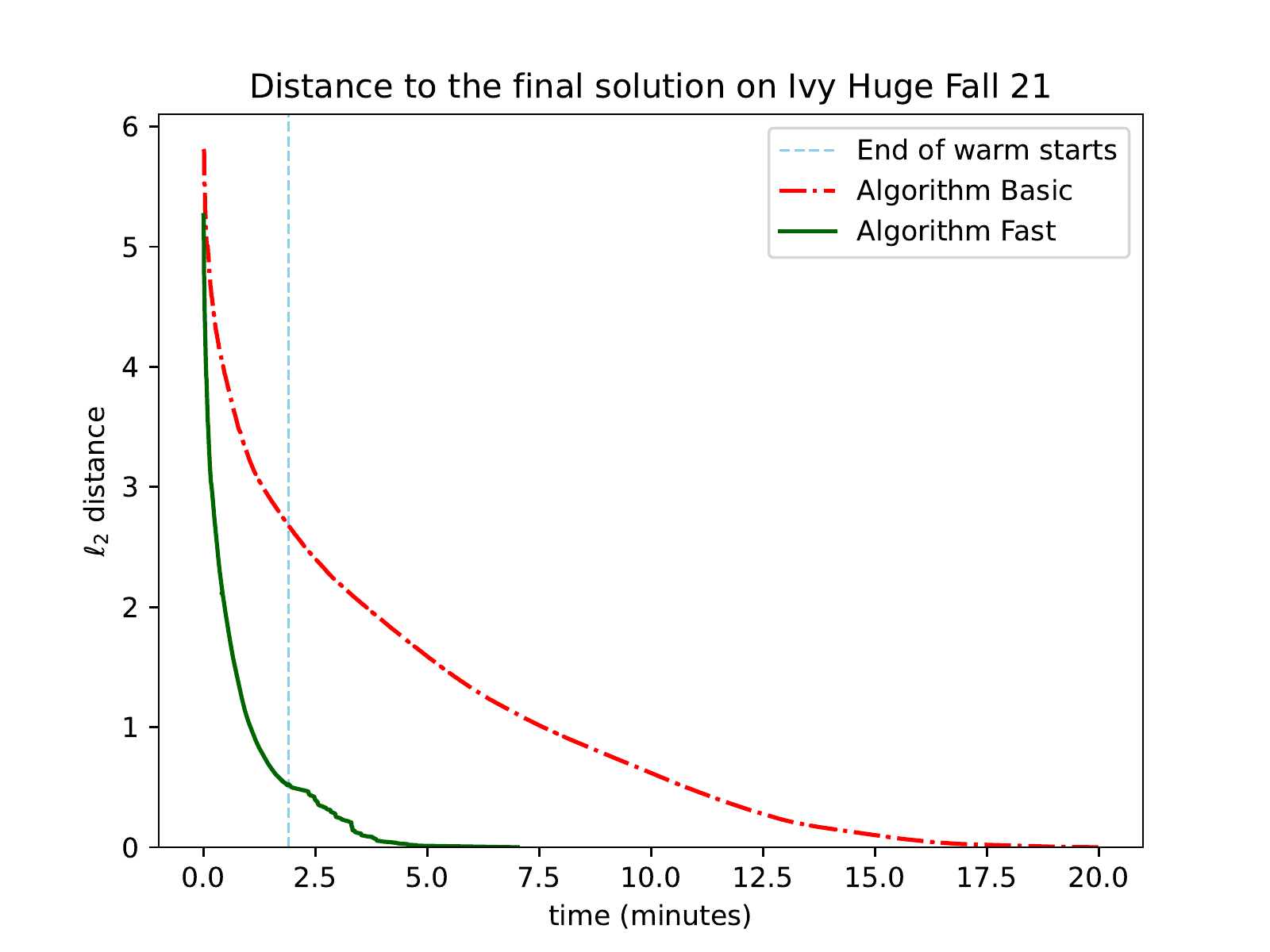}
    \includegraphics[width=.48\textwidth]{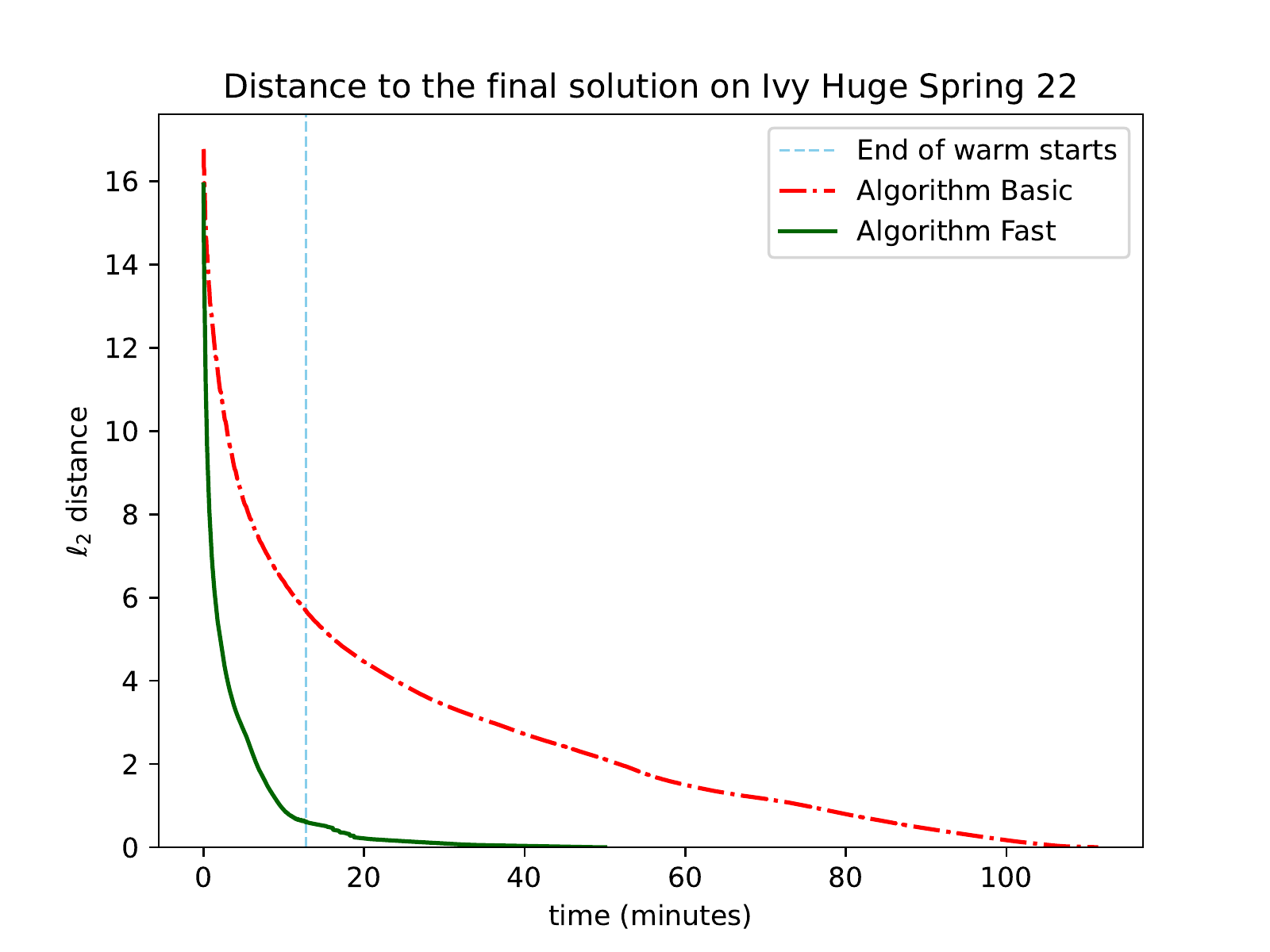}
    \\
    \includegraphics[width=.48\textwidth]{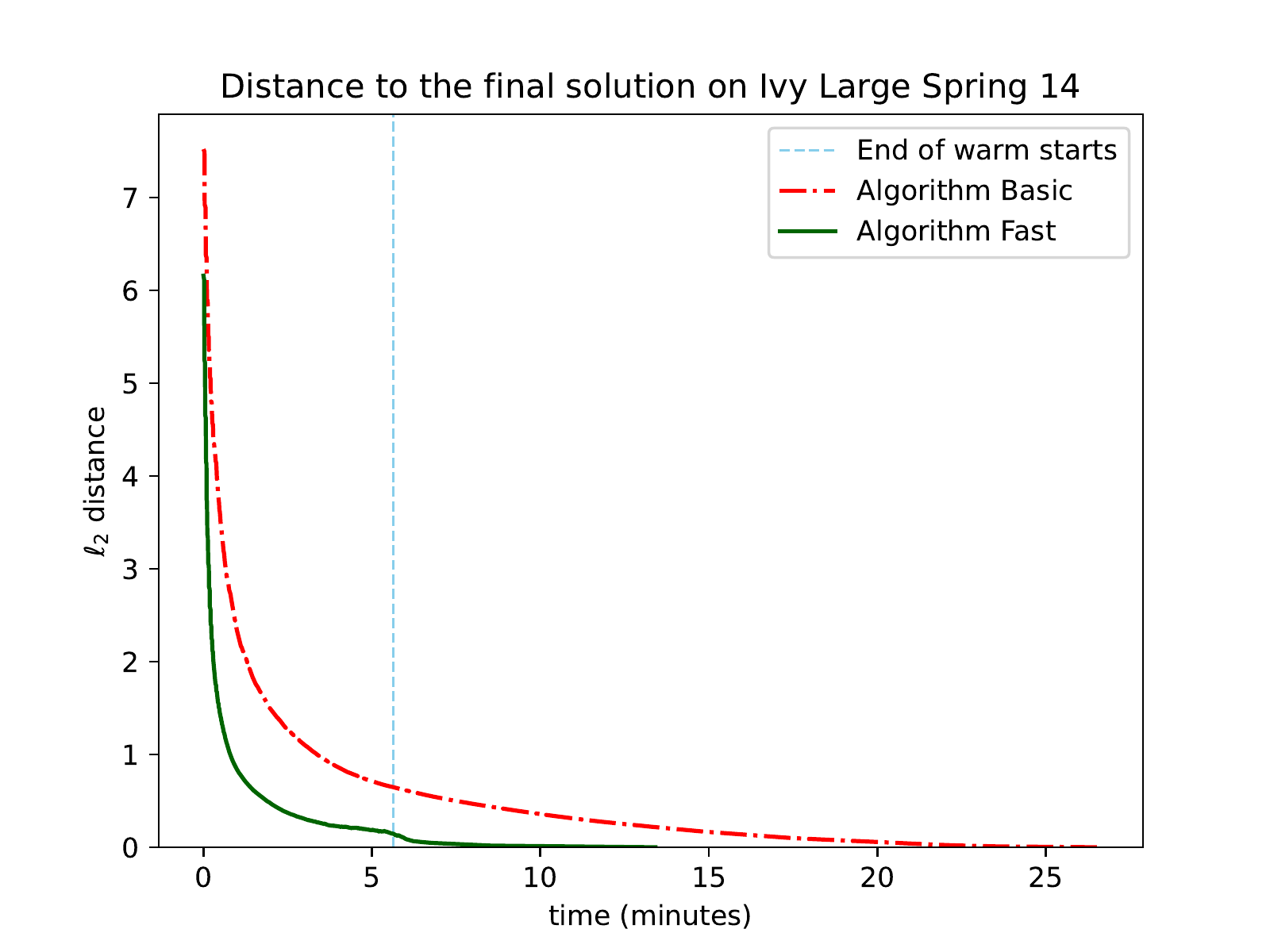}
    % \,
    % \\
    % \includegraphics[width=.49\textwidth]{img/ordata-distances-wrt-time.pdf}
    % \,
    % \includegraphics[width=.49\textwidth]{img/ordata-distances-wrt-iter.pdf}
    % \\
    % \includegraphics[width=.49\textwidth]{img/ordata-distances-wrt-time-2.pdf}
    % \,
    % \includegraphics[width=.49\textwidth]{img/ordata-distances-wrt-iter-2.pdf}
    % \\
    % \includegraphics[width=.49\textwidth]{img/giant-distances-wrt-time.pdf}
    % \,
    % \includegraphics[width=.49\textwidth]{img/giant-distances-wrt-iter.pdf}
    \caption{The distance between the found solution and the final solution with respect to the running time.}
    \label{fig:progress-of-algos-distance-app}
\end{figure}

\end{document}